\documentclass[a4paper,11pt]{article}

\usepackage{fullpage}
\usepackage{soul}
\usepackage[utf8]{inputenc}
\usepackage[small]{caption}
\usepackage{xcolor}
\usepackage{amsmath}
\usepackage{booktabs}
\usepackage{algorithm}
\usepackage{enumitem}
\usepackage{framed}
\usepackage{tikz}
\usepackage{url}

\usepackage{amsthm}
\usepackage{amssymb}

\newtheorem{theorem}{Theorem}[section]

\newtheorem{proposition}[theorem]{Proposition}
\newtheorem{lemma}[theorem]{Lemma}

\theoremstyle{definition}
\newtheorem{definition}[theorem]{Definition}

\DeclareMathOperator{\argmin}{argmin}
\DeclareMathOperator{\argmax}{argmax}
\newcommand{\mcut}{m_{\text{cut}}}
\newcommand{\Mcut}{M_{\text{cut}}}
\newcommand{\xx}{\mathbf{x}}

\usepackage{natbib}

\allowdisplaybreaks

\usepackage{authblk}

\title{\bf Consensus Halving for Sets of Items}

\author[1]{Paul W. Goldberg} 
\author[1]{Alexandros Hollender}
\author[2]{Ayumi Igarashi}
\author[3]{\\ Pasin Manurangsi}
\author[1]{Warut Suksompong}

\affil[1]{University of Oxford, UK}
\affil[2]{National Institute of Informatics, Japan}
\affil[3]{Google Research, USA}

\date{}

\begin{document}

\maketitle

\begin{abstract}
Consensus halving refers to the problem of dividing a resource into two parts so that every agent values both parts equally.
Prior work has shown that when the resource is represented by an interval, a consensus halving with at most $n$ cuts always exists, but is hard to compute even for agents with simple valuation functions.
In this paper, we study consensus halving in a natural setting where the resource consists of a set of items without a linear ordering.
When agents have additive utilities, we present a polynomial-time algorithm that computes a consensus halving with at most $n$ cuts, and show that $n$ cuts are almost surely necessary when the agents' utilities are drawn from probabilistic distributions.
On the other hand, we show that for a simple class of monotonic utilities, the problem already becomes PPAD-hard.
Furthermore, we compare and contrast consensus halving with the more general problem of consensus $k$-splitting, where we wish to divide the resource into $k$ parts in possibly unequal ratios, and provide some consequences of our results on the problem of computing small agreeable sets.
\end{abstract}

\section{Introduction}

Given a set of resources, how can we divide it between two families in such a way that every member of both families believes that the two resulting parts have the same value?
This is an important problem in resource allocation and has been addressed several times under different names \citep{Neyman46,HobbyRi65,Alon87}, with \emph{consensus halving} being the name by which it is best known today \citep{SimmonsSu03}.

In prior studies of consensus halving, the resource is represented by an interval, and the goal is to find an equal division into two parts that makes a small number of cuts in the interval.\footnote{\citet{SimmonsSu03} assume that the resource is a two- or three-dimensional object but only consider cuts by parallel planes; their model is therefore equivalent to that of a one-dimensional object.}
Using the Borsuk-Ulam theorem from topology, \citet{SimmonsSu03} established that for any continuous preferences of the $n$ agents involved, there is always a consensus halving that uses no more than $n$ cuts---this also matches the smallest number of cuts in the worst case.
In addition, the same authors developed an algorithm that computes an $\varepsilon$-approximate solution for any given $\varepsilon > 0$, meaning that the values of the two parts differ by at most $\varepsilon$ for every agent.
Although the algorithm is more efficient than a brute-force approach, its running time is exponential in the parameters of the problem.
This is in fact not a coincidence: \citet{FilosratsikasGo18} recently showed that $\varepsilon$-approximate consensus halving is PPA-complete, implying that the problem is unlikely to admit a polynomial-time algorithm.
\citet{FilosratsikasHoSo20} strengthened this result by proving that the problem remains hard even when the agents have simple valuations over the interval.
In particular, the PPA-completeness result holds for agents with ``two-block uniform'' valuations, i.e., valuation functions that are piecewise uniform over the interval and assign non-zero value to at most two separate pieces.

While these hardness results stand in contrast to the positive existence result, they rely crucially on the resource being in the form of an interval.
Most practical division problems do not fall under this assumption, including when we divide assets such as houses, cars, stocks, business ownership, or facility usage.
When each item is homogeneous, a consensus halving can be easily obtained by splitting every item in half.
However, since splitting individual assets typically involves an overhead, for example in managing a joint business or sharing the use of a house, we want to achieve a consensus halving while splitting only a small number of assets.
Fortunately, a consensus halving that splits at most $n$ items is guaranteed to exist regardless of the number of items---this can be seen by arranging the items on a line in arbitrary order and applying the aforementioned existence theorem of \citet{SimmonsSu03}.
The bound $n$ is also tight: if each agent only values a single item and the $n$ valued items are distinct, all of them clearly need to be split.
Nevertheless, given that the items do not inherently lie on a line, the hardness results from previous work do not carry over.
Could it be that computing a consensus halving efficiently is possible when the resource consists of a set of items?

\subsection{Overview of Results}

We assume throughout the paper that the resource is composed of $m$ items.
Each item is homogeneous, so the utility of an agent for a (possibly fractional) set of items depends only on the fractions of the $m$ items in that set.
For this overview we focus on the more interesting case where $n\leq m$, but all of our results can be extended to arbitrary $n$ and $m$.

We begin in Section~\ref{sec:additive} by considering agents with \emph{additive} utilities, i.e., the utility of each agent is additive across items and linear in the fraction of each item.
Under this assumption, we present a polynomial-time algorithm that computes a consensus halving with at most $n$ cuts by finding a vertex of the polytope defined by the relevant constraints.
This positive result stands in stark contrast with the PPA-hardness when the items lie on a line, which we obtain by discretizing an analogous hardness result of \citet{FilosratsikasHoSo20}.
We then show that improving the number of cuts beyond $n$ is difficult: even computing a consensus halving that uses at most $n-1$ cuts more than the minimum possible for a given instance is NP-hard.
Nevertheless, we establish that instances admitting a solution with fewer than $n$ cuts are rare.
In particular, if the agents' utilities for items are drawn independently from non-atomic distributions, it is almost surely the case that every consensus halving requires no fewer than $n$ cuts.

Next, in Section~\ref{sec:monotonic}, we address the broader class of \emph{monotonic} utilities, wherein an agent's utility for a set does not decrease when any fraction of an item is added to the set.
For such utilities, we show that the problem of computing a consensus halving with at most $n$ cuts becomes PPAD-hard, thereby providing strong evidence of its computational hardness.\footnote{We refer to \citep[Chapter~20]{Roughgarden16} for a discussion of the complexity class PPAD.}
Perhaps surprisingly, this hardness result holds even for the class of utility functions that we call ``symmetric-threshold utilities'', which are very close to being additive.
Indeed, such utility functions are additive across items; for each item, having a sufficiently small fraction of the item is the same as not having the item at all, having a sufficiently large fraction of it is the same as having the whole item, and the utility increases linearly in between.
On the other hand, we present a number of positive results for monotonic utilities when the number of agents is constant in Appendix~\ref{app:constant-n}.

In Section~\ref{sec:agreeable}, we provide some implications of our results on the ``agreeable sets'' problem studied by \citet{ManurangsiSu19}.
A set is said to be \emph{agreeable} to an agent if the agent likes it at least as much as the complement set.
Manurangsi and Suksompong proved that a set of size at most $\left\lfloor\frac{m+n}{2}\right\rfloor$ that is agreeable to all agents always exists, and this bound is tight.
They then gave polynomial-time algorithms that compute an agreeable set matching the tight bound for two and three agents.
We significantly generalize this result by exhibiting efficient algorithms for any number of agents with additive utilities, as well as any \emph{constant} number of agents with monotonic utilities.
In addition, we present a short alternative proof for the bound $\left\lfloor\frac{m+n}{2}\right\rfloor$ via consensus halving.

Finally, in Section~\ref{sec:k-splitting}, we study the more general problem of \emph{consensus $k$-splitting} for agents with additive utilities.
Our aim in this problem is to split the items into $k$ parts so that all agents agree that the parts are split according to some given ratios $\alpha_1,\dots,\alpha_k$; consensus halving corresponds to the special case where $k=2$ and $\alpha_1=\alpha_2=1/2$.
Unlike for consensus halving, however, in consensus $k$-splitting we may want to cut the same item more than once when $k>2$, so we cannot assume without loss of generality that the number of cuts is equal to the number of items cut.
For any $k$ and any ratios $\alpha_1,\dots,\alpha_k$, we show that there exists an instance in which cutting $(k-1)n$ items is necessary.
On the other hand, a generalization of our consensus halving algorithm from Section~\ref{sec:additive} computes a consensus $k$-splitting with at most $(k-1)n$ cuts in polynomial time, thereby implying that the bound $(k-1)n$ is tight for both benchmarks.
We also illustrate further differences between consensus $k$-splitting and consensus halving, both with respect to item ordering and from the probabilistic perspective.

\subsection{Related Work}

Consensus halving falls under the broad area of \emph{fair division}, which studies how to allocate resources among interested agents in a fair manner \citep{BramsTa96,BramsTa99,Moulin03}. 
Common fairness notions include \emph{envy-freeness}---no agent envies another agent in view of the bundles they receive---and \emph{equitability}---all agents have the same utility for their own bundle.
The fair division literature typically assumes that each recipient of a bundle is either a single agent or a group of agents represented by a single preference.
However, a number of recent papers have considered an extension of the traditional setting to groups, thereby allowing us to capture the differing preferences within the same group as in our introductory example with families \citep{ManurangsiSu17,Suksompong18,KyropoulouSuVo19,SegalhaleviNi19,SegalhaleviSu19,SegalhaleviSu20}.
Note that a consensus halving is envy-free for all members of the two groups; moreover, it is equitable provided that the utilities of the agents are additive and normalized so that every agent has the same value for the entire set of items.

A classical fair division algorithm that dates back over two decades is the \emph{adjusted winner procedure}, which computes an envy-free and equitable division between two agents \citep{BramsTa96}.\footnote{See http://www.nyu.edu/projects/adjustedwinner for a demonstration and implementation of the procedure.}
The procedure has been suggested for resolving divorce settlements and international border disputes, with one of its advantages being the fact that it always splits at most one item.
\citet{SandomirskiySe19} investigated the problem of attaining fairness while minimizing the number of shared items, and gave algorithms and hardness results for several variants of the problem.
Like in our work, both the adjusted winner procedure and the work of \citet{SandomirskiySe19} assume that items are homogeneous and, as in Section~\ref{sec:additive}, that the agents' utilities are linear in the fraction of each item and additive across items.
Moreover, both of them require the assumption that all items can be shared---if some items are indivisible, then an envy-free or equitable allocation cannot necessarily be obtained.\footnote{This motivates relaxations such as \emph{envy-freeness up to one item (EF1)} and \emph{envy-freeness up to any item (EFX)}, which have been extensively studied in the last few years (e.g., \citep{CaragiannisKuMo19,PlautRo20}). However, as \citet{SandomirskiySe19} noted, when a divorcing couple decides how to split their children or two siblings try to divide three houses between them, it is unlikely that anyone will agree to a bundle that is envy-free up to one child or house.}

Besides consensus halving, another problem that also involves dividing items into equal parts is \emph{necklace splitting}, which can be seen as a discrete analog of consensus halving
\citep{GoldbergWe85,AlonWe86,Alon87}.
In a basic version of necklace splitting, there is a necklace with beads of $n$ colors, with each color having an even number of beads.
Our task is to split the necklace using at most $n$ cuts and arrange the resulting pieces into two parts so that the beads of each color are evenly distributed between both parts.
Observe that the difficulty of this problem lies in the spatial ordering of the beads---the problem would be trivial if the beads were unordered items as in our setting.
While consensus halving and necklace splitting have long been studied by mathematicians, they recently gained significant interest among computer scientists thanks in large part to new computational complexity results \citep{FilosratsikasFrGo18,FilosratsikasGo18,FilosratsikasGo19,DeligkasFeMe19,AlonGr20,FilosratsikasHoSo20,FilosratsikasHoSo20b}.
In particular, the PPA-completeness result of \citet{FilosratsikasGo18} for approximate consensus halving was the first such result for a problem that is ``natural'' in the sense that its description does not involve a polynomial-sized circuit.

\section{Additive Utilities}
\label{sec:additive}

We first formally define the problem of consensus halving for a set of items.
There is a set $N=[n]$ of $n$~agents and a set $M=[m]$ of $m$ items, where $[r]:=\{1,2,\dots,r\}$ for any positive integer $r$.
A \emph{fractional set of items} contains a fraction $x_j\in[0,1]$ of each item $j$.
We will mostly be interested in fractional sets of items in which only a small number of items are fractional---that is, most items have $x_j=0$ or $1$.
Agent $i$ has a utility function $u_i$ that describes her nonnegative utility for any fractional set of items; for an item $j\in M$, we sometimes write $u_i(j)$ to denote $u_i(\{j\})$.
A \emph{partition of $M$ into fractional sets of items} $M_1,\dots,M_k$ has the property that for every item $j\in M$, the fractions of item $j$ in the $k$ fractional sets sum up to exactly $1$.

\begin{definition}
\label{def:consensus-halving}
A \emph{consensus halving} is a partition of $M$ into two fractional sets of items $M_1$ and $M_2$ such that $u_i(M_1) = u_i(M_2)$ for all $i\in N$.
An item is said to be \emph{cut} if there is a positive fraction of it in both parts of the partition.
\end{definition}

In this section, we assume that the agents' utility functions are \emph{additive}.
This means that for a set $M'$ containing a fraction $x_j$ of item $j$, the utility of agent $i$ is given by $u_i(M') = \sum_{j\in M}x_j\cdot u_i(j)$.
Observe that under additivity, $M'$ forms one part of a consensus halving exactly when 
\begin{align} \label{eq:halving-constraint}
\sum_{j\in M}x_j\cdot u_i(j) = \frac{1}{2}\sum_{j\in M} u_i(j) & & \forall i \in N.
\end{align}
As we mentioned in the introduction, a consensus halving with no more than $n$ cuts is guaranteed to exist regardless of the number of items.
Our first result shows that such a division can be found efficiently for additive utilities.

\begin{theorem}
\label{thm:polytope}
For $n$ agents with additive utilities, there exists a polynomial-time algorithm that computes a consensus halving with at most $\min\{n,m\}$ cuts.
\end{theorem}

\begin{proof}
If $n\geq m$, a partition that divides every item in half is clearly a consensus halving and makes $m=\min\{n,m\}$ cuts. 
We therefore assume from now on that $n\leq m$ and describe a polynomial-time algorithm that computes a consensus halving using no more than $n$ cuts.

The main idea of our algorithm is to start with the trivial consensus halving where $x_1 = x_2 = \dots = x_m = 1/2$, and then gradually reduce the number of cuts. 
We stop when the process cannot be continued, at which point we show that the consensus halving must contain at most $n$ cuts.
Our algorithm is presented below.
\begin{enumerate}
\item Let $x_1 = x_2 = \dots = x_m = 1/2$.
\item Let $S$ denote the set of $n$ equations $\sum_{j\in M}\left(y_j - \frac{1}{2}\right)\cdot u_i(j) = 0$ for $i \in N$, and let $T = \emptyset$.
\item While there exists a solution $(y_1, \dots, y_m) \ne (x_1, \dots, x_m)$ to $S \cup T$, do the following: \label{step:gaussian-elimination}
\begin{enumerate}
\item For every $j \in M$ such that $y_j \ne x_j$, compute
\begin{align*}
\gamma_j :=
\begin{cases}
\frac{1 - x_j}{y_j - x_j} & \text{ if } y_j > x_j; \\
\frac{x_j}{x_j - y_j} & \text{ if } y_j < x_j.
\end{cases}
\end{align*}
\item Let $j^* = \argmin_{j \in M, y_j \ne x_j} \gamma_j$.
\item For every $j \in M$, let $s_j := (1 - \gamma_{j^*}) \cdot x_j + \gamma_{j^*} \cdot y_j$, and update the value of $x_j$ to $s_j$. \label{step:linear-comb}
\item Add the equation $y_{j^*} = x_{j^*}$ to $T$.
\end{enumerate}
\item Output $(x_1, \dots, x_m)$.
\end{enumerate}
Finding a solution $(y_1,\dots,y_m)$ to $S \cup T$ that is not equal to $(x_1, \dots, x_m)$ or determining that such a solution does not exist (Step~\ref{step:gaussian-elimination}) can be done in polynomial time via Gaussian elimination.\footnote{Specifically, if the linear equations in $S \cup T$ lead to a unique solution $(x_1, \dots, x_m)$, then Gaussian elimination immediately results in this solution. Otherwise, Gaussian elimination will yield a row echelon form; by setting one of the non-pivots $y_j$ to be an arbitrary number not equal to $x_j$, we obtain a solution that is not equal to $(x_1, \dots, x_m)$.}
Moreover, it is obvious that the other steps of the algorithm run in polynomial time.

We next prove the correctness of our algorithm, starting with arguing that $(x_1, \dots, x_m)$ forms a consensus halving. Since we start with a consensus halving $x_1 = \cdots = x_m = 1/2$, it suffices to show that each execution of the loop in Step~\ref{step:gaussian-elimination} preserves the validity of the solution. Observe that, since both $(x_1, \dots, x_m)$ and $(y_1, \dots, y_m)$ are solutions to the equations~\eqref{eq:halving-constraint}, their convex combination (in Step~\ref{step:linear-comb}) also satisfies the equations~\eqref{eq:halving-constraint}. 
Furthermore, for each $j$ such that $y_j\neq x_j$, the value $\gamma_j$ is chosen so that if we replace $\gamma_{j^*}$ by $\gamma_j$ in the formula for $s_j$, we would have $s_j=1$ for the case $y_j>x_j$, and $s_j=0$ for the case $y_j<x_j$.
Since $\gamma_{j^*}\leq \gamma_j$, we have that $s_j\in[0,1]$ for all $j$ such that $y_j\neq x_j$.
In addition, the value of $x_j$ does not change for $j$ such that $y_j=x_j$.
Thus, $(x_1, \dots, x_m)$ remains a consensus halving throughout the algorithm.

Finally, we are left to show that at most $n$ items are cut in the output $(x_1, \dots, x_m)$. 
As noted above, our definition of $\gamma_j$ ensures that $x_{j^*} \in \{0, 1\}$ after the execution of Step~\ref{step:linear-comb}. Furthermore, as the constraint $y_{j^*} = x_{j^*}$ is then immediately added to $T$, the value of $x_{j^*}$ does not change for the rest of the algorithm. 
As a result, every item $j \in T$ is uncut. Thus, it suffices to show that $|T| \geq m - n$ at the end of the execution.

When the while loop in Step~\ref{step:gaussian-elimination} terminates, $(x_1, \dots, x_m)$ must be the unique solution to $S \cup T$. Recall that a system of linear equations with $m$ variables can only have a unique solution when the number of constraints is at least $m$. This means that $|S \cup T| \geq m$ at the end of the algorithm. Since $|S| = n$, we must have $|T| \geq m - n$, as desired. 
\end{proof}

Note that the above algorithm can be viewed as finding a vertex of the polytope defined by the constraints~\eqref{eq:halving-constraint} and $0 \leq x_j \leq 1$ for all $j \in M$. 
In fact, it suffices to use a generic algorithm for this task; however, to the best of our knowledge, such algorithms often involve solving a linear program, whereas the algorithm presented above is conceptually simple and can be implemented directly.
We also remark that our algorithm works even when some utilities $u_i(j)$ are negative, i.e., some of the items are goods while others are chores.
Allocating a combination of goods and chores has received increasing attention in the fair division community \citep{BogomolnaiaMoSa17,Segalhalevi18,AzizCaIg19}.

As we discussed in the introduction, an important reason behind the positive result in Theorem~\ref{thm:polytope} is the lack of linear order among the items.
Indeed, as we show next, if the items lie on a line and we are only allowed to cut the line using $n$ cuts, finding a consensus halving becomes computationally hard.
This follows from discretizing the hardness result of \citet{FilosratsikasHoSo20} and holds even if we allow the consensus halving to be approximate instead of exact.
Formally, when the items lie on a line, we may place a number of cuts, with each cut lying either between two adjacent items or at some position within an item.
All (fractional or whole) items between any two adjacent cuts must belong to the same fractional set of items in a partition, where the left and right ends of the line also serve as ``cuts'' in this requirement (see Figure~\ref{fig:consensus-halving-line} for an example).
We say that a partition into fractional sets of items $(M_1,M_2)$ is an \emph{$\varepsilon$-approximate consensus halving} if $|u_i(M_1)-u_i(M_2)|\leq \varepsilon\cdot u_i(M)$ for every agent~$i$.

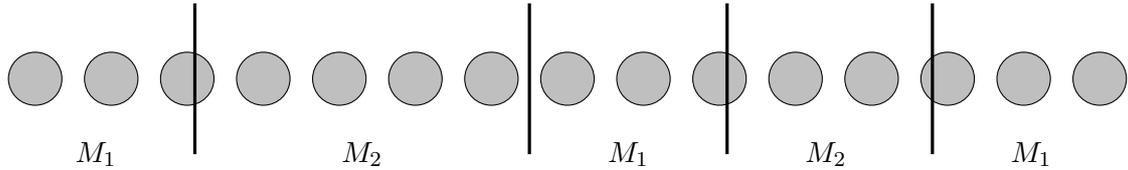
\begin{figure}
\centering
\begin{tikzpicture}

\filldraw[fill=lightgray] (0,0) circle (10pt);
\filldraw[fill=lightgray] (1,0) circle (10pt);
\filldraw[fill=lightgray] (2,0) circle (10pt);
\filldraw[fill=lightgray] (3,0) circle (10pt);
\filldraw[fill=lightgray] (4,0) circle (10pt);
\filldraw[fill=lightgray] (5,0) circle (10pt);
\filldraw[fill=lightgray] (6,0) circle (10pt);
\filldraw[fill=lightgray] (7,0) circle (10pt);
\filldraw[fill=lightgray] (8,0) circle (10pt);
\filldraw[fill=lightgray] (9,0) circle (10pt);
\filldraw[fill=lightgray] (10,0) circle (10pt);
\filldraw[fill=lightgray] (11,0) circle (10pt);
\filldraw[fill=lightgray] (12,0) circle (10pt);
\filldraw[fill=lightgray] (13,0) circle (10pt);
\filldraw[fill=lightgray] (14,0) circle (10pt);

\draw[very thick] (2.1,1) -- (2.1,-1);
\draw[very thick] (6.5,1) -- (6.5,-1);
\draw[very thick] (9.1,1) -- (9.1,-1);
\draw[very thick] (11.8,1) -- (11.8,-1);

\node at (0.8,-1) {$M_1$};
\node at (4.3,-1) {$M_2$};
\node at (7.8,-1) {$M_1$};
\node at (10.4,-1) {$M_2$};
\node at (13.1,-1) {$M_1$};

\end{tikzpicture}
\caption{Consensus halving for items on a line: in this example there are $15$ items (represented by gray balls) that lie on a line and we have used $4$ cuts to obtain a partition into fractional sets of items $(M_1,M_2)$. The labels $M_1$ and $M_2$ indicate the set to which each segment belongs.}
\label{fig:consensus-halving-line}
\end{figure}

\begin{theorem}
\label{thm:PPA-line}
Suppose that the items lie on a line. There exists a polynomial $p$ such that finding a $1/p(n)$-approximate consensus halving for $n$ agents with at most $n$ cuts on the line is PPA-hard, even if the valuations are binary and every agent values at most two contiguous blocks of items.
\end{theorem}

\begin{proof}
We prove this by discretizing the hard instances constructed by \citet[Theorem~2]{FilosratsikasHoSo20}. 
In their setting there are $n$ agents who have piecewise-uniform valuation functions $v_1, \dots, v_n$ over the interval $[0,1]$.\footnote{This means that for each agent $i$, the interval $[0,1]$ can be partitioned into a finite number of intervals so that the density of the agent's valuation function is either $0$ or some constant $c_i$ over each interval.} 
By a closer inspection of their proof, we note that the instances they construct have some useful properties. Namely, there exist polynomials $p$ and $q$ such that:
\begin{enumerate}
	\item Every agent has a two-block uniform valuation on $[0,1]$, i.e., the density of the valuation function is piecewise-uniform and non-zero in at most two intervals. In other words, every agent has (at most) two blocks of value and they have the same height.
	\item There exists an integer $d \leq q(n)$ such that for all agents, the endpoints of the blocks are rational numbers with denominator $d$.
	\item Finding a $1/p(n)$-approximate consensus halving is PPA-hard.
\end{enumerate}

Using these properties, we can construct an equivalent instance in our setting. We position $m=d$ items on a line, where the $j$th item represents the interval $I_j := [(j-1)/d,j/d]$ in the original instance. 
Note that for every agent of the original instance, the density of their valuation function is constant over $I_j$ for each $j$. Thus, by letting
\begin{equation*}
u_i(j) = \left\{\begin{tabular}{ll}
$1$ & if $v_i((j-1)/d,j/d) > 0$; \\
$0$ & if $v_i((j-1)/d,j/d) = 0$
\end{tabular}
\right.
\end{equation*}
for all $i \in [n]$ and $j \in [d]$, we have exactly recreated the same valuation functions in our setting, up to normalization. 
In particular, any $1/p(n)$-approximate consensus halving of the items using at most $n$ cuts on the line immediately yields a $1/p(n)$-approximate consensus halving of $v_1, \dots, v_n$ using at most $n$ cuts on $[0,1]$, implying that our problem is also PPA-hard. 
\end{proof}

Although Theorem~\ref{thm:polytope} allows us to efficiently compute a consensus halving with no more than $n$ cuts in any instance, for some instances there exists a solution using fewer cuts.
An extreme example is when all agents have the same utility function, in which case a single cut already suffices.
This raises the question of determining the least number of cuts required for a given instance.
Unfortunately, when there is a single agent, deciding whether there is a consensus halving that leaves all items uncut is already equivalent to the well-known NP-hard problem \textsc{Partition}.
For general $n$, even computing a division that uses at most $n-1$ cuts more than the optimal solution is still computationally hard, as the following theorem shows.

\begin{theorem}
\label{thm:NP-partition}
For $n$ agents with additive utilities, it is NP-hard to compute a consensus halving that uses at most $n-1$ cuts more than the minimum number of cuts for the same instance.
\end{theorem}

\begin{proof}
We reduce from the NP-hard problem \textsc{Partition}. 
Let $w_1, \dots, w_r$ be the integers that form a \textsc{Partition} instance. 
We construct a consensus halving instance $I$ with $n$ agents and a set of $n \times r$ items $M=\{(\ell,j) : \ell \in [n], j \in [r]\}$. Every agent values a distinct set of items according to the numbers $w_1, \dots, w_r$. Formally,
\begin{equation*}
u_i((\ell,j)) = \left\{\begin{tabular}{ll}
$w_j$ & if $\ell=i$; \\
$0$ & if $\ell \neq i$
\end{tabular}
\right.
\end{equation*}
for all $i,\ell \in [n]$ and $j \in [r]$. 
It is easy to see that this instance has the following properties:
\begin{enumerate}
	\item If $w_1, \dots, w_r$ can be partitioned into two sets of equal sum, then our instance $I$ admits a consensus halving using no cut.
	\item If $w_1, \dots, w_r$ cannot be partitioned into two sets of equal sum, then any consensus halving of our instance $I$ uses at least $n$ cuts. 
	This is because in that case, for every agent $i \in N$, at least one of the items $(i,1), \dots, (i,r)$ must be cut.
\end{enumerate}
As a result, in the first case, any consensus halving that uses at most $n-1$ cuts more than the minimum number of cuts will have at most $n-1$ cuts. 
In the second case, any consensus halving that uses at most $n-1$ cuts more than the minimum number of cuts will have at least $n$ cuts. 
Thus, \textsc{Partition} reduces to the problem of computing a consensus halving that uses at most $n-1$ cuts more than the minimum number of cuts.
\end{proof}

Theorem~\ref{thm:NP-partition} implies that there is no hope of finding a consensus halving with the minimum number of cuts or even a non-trivial approximation thereof in polynomial time, provided that P $\neq$ NP.
Nevertheless, we show that instances that admit a consensus halving with fewer than $n$ cuts are rare: if the utilities are drawn independently at random from probability distributions, then it is almost surely the case that any consensus halving needs at least $n$ cuts.
We say that a distribution is \emph{non-atomic} if it does not put positive probability on any single point.

\begin{theorem}
\label{thm:asymptotic-n-cuts}
Suppose that for each $i\in N$ and $j\in M$, the utility $u_i(j)$ is drawn independently from a non-atomic distribution $\mathcal{D}_{i,j}$.
Then, with probability $1$, every consensus halving uses at least $\min\{n,m\}$ cuts.
\end{theorem}

\begin{proof}
The high-level idea is to show that if there are less than $\min\{n,m\}$ cuts, then a certain utility $u_i(j)$ needs to take on a specific value; this event occurs with probability $0$ since the distribution $\mathcal{D}_{i,j}$ is non-atomic.

Let $m_{\text{cut}} = \min\{n, m\} - 1$. 
Recall that a consensus halving corresponds to a tuple $(x_1, \dots, x_m) \in [0, 1]^m$ for which the constraint~\eqref{eq:halving-constraint} is satisfied, and that item $j$ is cut if and only if $x_j \notin \{0, 1\}$. 
As a result, from union bound, it suffices to show that for any fixed $M_{\text{cut}} \subseteq M$ of size $m_{\text{cut}}$, we have
\begin{align}
\Pr[\exists (x_1, \dots, x_m) \in [0, 1]^m \text{ that satisfies } \eqref{eq:halving-constraint} \text{ and } x_j \in \{0, 1\} \text{ for all }  j \notin M_{\text{cut}}] = 0. \label{eq:first-union-bound-asymptotic}
\end{align}
For notational convenience, we will only show that~\eqref{eq:first-union-bound-asymptotic} holds for $\Mcut = \{1, \dots, \mcut\}$; due to symmetry, the same bound also holds for every $\Mcut \subseteq M$ of size $\mcut$.

To show~\eqref{eq:first-union-bound-asymptotic} for $\Mcut = \{1, \dots, \mcut\}$, we may apply the union bound again to derive
\begin{align*}
&\Pr[\exists (x_1, \dots, x_m) \in [0, 1]^m \text{ that satisfies } \eqref{eq:halving-constraint} \text{ and } x_j \in \{0, 1\} \text{ for all }  j  \in \{\mcut + 1, \dots, m\}] \\
&\leq \sum_{t_{\mcut + 1}, \dots, t_{m} \in \{0, 1\}} \Pr[\exists x_1, \dots, x_{\mcut} \in  [0, 1] \text{ such that } (x_1, \dots, x_{\mcut}, t_{\mcut + 1}, \dots, t_m) \text{ satisfies } \eqref{eq:halving-constraint}].
\end{align*}
Hence, it suffices to show that, for any fixed $t_{\mcut + 1}, \dots, t_{m} \in \{0, 1\}$, we have 
\[
\Pr[\exists x_1, \dots, x_{\mcut} \in  [0, 1] \text{ such that } (x_1, \dots, x_{\mcut}, t_{\mcut + 1}, \dots, t_m) \text{ satisfies } \eqref{eq:halving-constraint}] = 0.
\]

To see that this is the case, consider any fixed values of $u_i(j)$ for all $i \in N, j \in \Mcut$; we will show that the above probability is $0$ over the randomness of the utilities $u_i(j)$ for $i\in N, j\notin \Mcut$.
We may rearrange the constraint~\eqref{eq:halving-constraint} as
\begin{align}
\sum_{j\in \Mcut} u_i(j) \cdot x_j = \frac{1}{2}\sum_{j\in \Mcut} u_i(j) + \sum_{j \notin \Mcut} \left(\frac{1}{2} - t_j\right) \cdot u_i(j) & & \forall i \in N. \label{eq:halving-constraint-rearranged}
\end{align}
Now, since there are $n$ linear equations and only $\mcut < n$ variables $x_1, \dots, x_{\mcut}$, the coefficient vectors $(u_1(1),\dots,u_1(\mcut)),\dots,(u_n(1),\dots,u_n(\mcut))$ must be linearly dependent. 
In other words, there exists $(a_1, \dots, a_n) \ne (0, \dots, 0)$ such that
\begin{align*}
\sum_{i \in N} a_i \cdot u_i(j) = 0 & & \forall j \in \Mcut.
\end{align*}
Hence, by taking the corresponding linear combination of~\eqref{eq:halving-constraint-rearranged}, we have
\begin{align*}
0 
&= \sum_{j\in \Mcut} x_j \left(\sum_{i\in N} a_i\cdot u_i(j)\right)  \\
&= \sum_{i\in N} a_i \left(\sum_{j\in \Mcut} x_j\cdot u_i(j)\right)  \\
&= \sum_{i \in N} a_i\left(\frac{1}{2} \cdot \sum_{j\in \Mcut} u_i(j) + \sum_{j \notin \Mcut} \left(\frac{1}{2} - t_j\right) \cdot u_i(j)\right).
\end{align*}
From $(a_1, \dots, a_n) \ne (0, \dots, 0)$, there exists $i^* \in N$ such that $a_{i^*} \ne 0$. 
Moreover, since $\mcut < m$, we have $m\notin\Mcut$.
The above equality therefore implies that
\begin{align*}
u_{i^*}(m) &= \frac{1}{\left(t_m - \frac{1}{2}\right)} \Bigg( \sum_{i \in N \atop i \ne i^*} \frac{a_i}{a_{i^*}}\left(\frac{1}{2} \cdot \sum_{j\in \Mcut} u_i(j) + \sum_{j \notin \Mcut} \left(\frac{1}{2} - t_j\right) \cdot u_i(j)\right) \\
& \qquad \qquad \qquad + \left(\frac{1}{2} \cdot \sum_{j\in \Mcut} u_{i^*}(j) + \sum_{j \notin \Mcut \atop j\ne m} \left(\frac{1}{2} - t_j\right) \cdot u_{i^*}(j)\right)\Bigg),
\end{align*}
where $t_m-1/2$ is nonzero because $t_m\in\{0,1\}$.
Since $\mathcal{D}_{i^*, m}$ is non-atomic and the utilities are drawn independently, the above equality occurs with probability 0, which implies that 
\begin{align*}
\Pr[\exists x_1, \dots, x_{\mcut} \in  [0, 1] \text{ such that } (x_1, \dots, x_{\mcut}, t_{\mcut + 1}, \dots, t_m) \text{ satisfies } \eqref{eq:halving-constraint}] = 0.
\end{align*}
As discussed, this in turn implies that the probability that there is a consensus halving with at most $\mcut$ cuts is $0$, concluding our proof.
\end{proof}

We now comment on the necessity of the two distributional assumptions in Theorem~\ref{thm:asymptotic-n-cuts}.
\begin{itemize}
\item Non-atomicity condition: Suppose $n=1$ and $\mathcal{D}_{1,j}$ is the Bernoulli distribution with $p=1/2$ for all $j\in M$, i.e., $u_1(j)=0$ and $u_1(j)=1$ with probability $1/2$ each.
Then the minimum number of cuts is $1$ if $u_i(j)=1$ for an odd number of $j$, and $0$ otherwise; the probability that each event occurs is $1/2$.
\item Independence condition: Suppose all agents have the same utility function, i.e., the dependence between the utilities is such that
$u_1(j)=\dots=u_n(j)$ for all $j\in[m]$.
In this case, it is clear that no more than one cut is needed regardless of $n$ and $m$.
\end{itemize}

As our final remark of this section, consider utility functions that are again additive across items, but for which the utility of each item scales \emph{quadratically} as opposed to linearly in the fraction of the item.
That is, for a set $M'$ containing a fraction $x_j$ of item $j$, the utility of agent $i$ is given by $u_i(M') = \sum_{j\in M}x_j^2\cdot u_i(j)$.
Even though these utility functions appear different from the ones we have considered so far, it turns out that the set of consensus halvings remains exactly the same.
Indeed, a partition $(M_1,M_2)$ is a consensus halving under the quadratic functions if and only if
\begin{align*}
\sum_{j\in M}x_j^2\cdot u_i(j) = \sum_{j\in M}(1-x_j)^2\cdot u_i(j) & & \forall i \in N.
\end{align*}
Since $x_j^2-(1-x_j)^2 = x_j-(1-x_j) = 2x_j - 1$, the above condition is equivalent to \eqref{eq:halving-constraint}, so all of our results in this section apply to the quadratic functions as well.

\section{Monotonic Utilities}
\label{sec:monotonic}

Next, we turn our attention to utility functions that are no longer additive as in Section~\ref{sec:additive}.
We assume that the utilities are \emph{monotonic}, meaning that the utility of an agent for a set of items cannot decrease upon adding any fraction of an item to the set.
Our main result is that finding a consensus halving is computationally hard for such valuations; in fact, the hardness holds even when the utilities take on a specific structure that we call \emph{symmetric-threshold}. Symmetric-threshold utilities are additive over items, and linear with symmetric thresholds within every item. Formally, the utility of agent $i$ for a fractional set of items $M'$ containing a fraction $x_j \in [0,1]$ of each item $j$ can be written as $u_i(M') = \sum_{j\in M} f_{ij}(x_j) \cdot u_i(j)$, where
\vspace{5mm}

\begin{minipage}{0.45\textwidth}
\begin{equation*}
f_{ij}(x_j) := \left\{\begin{tabular}{cl}
$0$ & if $x_j \leq c_{ij}$;\\
$\frac{x_j-c_{ij}}{1-2c_{ij}}$ & if $c_{ij} < x_j < 1-c_{ij}$;\\
$1$ & if $x_j \geq 1 - c_{ij}$,
\end{tabular}\right.
\end{equation*}
\end{minipage}
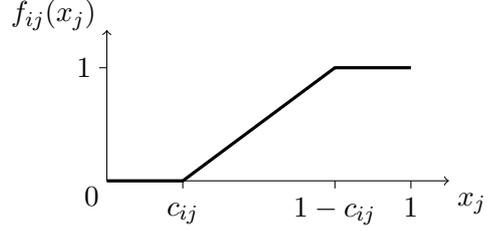
\begin{minipage}{0.5\textwidth}
\centering
\begin{tikzpicture}	
\draw[->] (0,0) -- (4.5,0);
\draw[->] (0,0) -- (0,2);
\draw[very thick] (0,0) -- (1,0) -- (3,1.5) -- (4,1.5);

\node at (4.8,-0.3) {$x_j$};
\node at (-0.7,2.2) {$f_{ij}(x_j)$};

\draw (-0.1,1.5) -- (0,1.5);
\node at (-0.3,1.5) {$1$};

\draw (1,0) -- (1,-0.1);
\node at (1,-0.43) {$c_{ij}$};

\draw (3,0) -- (3,-0.1);
\node at (3,-0.4) {$1-c_{ij}$};

\draw (4,0) -- (4,-0.1);
\node at (4,-0.35) {$1$};

\node at (-0.2,-0.2) {$0$};
		
\end{tikzpicture}
\end{minipage}
\vspace{5mm}

\noindent where $c_{ij} \in [0,1/2)$ is the \emph{threshold} or \emph{cap} of agent $i$ for item $j$. Intuitively, symmetric-threshold utilities model settings where having a small fraction of an item is the same as not having the item at all, while having a large fraction of the item is the same as having the whole item. The point where this threshold behavior occurs is controlled by the cap $c_{ij}$, which can be different for every pair $(i,j) \in N \times M$. It is easy to see that the resulting utility functions are indeed monotonic.
Note that although general monotonic utility functions do not necessarily admit a concise representation (see the discussion preceding Theorem~\ref{thm:agreeable-algo}), symmetric-threshold utility functions can be described succinctly.

Even though symmetric-threshold utility functions are very close to being additive, we show that finding a consensus halving for such utilities is computationally hard. Recall that a partition $(M_1,M_2)$ is an $\varepsilon$-approximate consensus halving if $|u_i(M_1)-u_i(M_2)|\leq \varepsilon\cdot u_i(M)$ for every agent~$i$.

\begin{theorem}
\label{thm:PPAD}
There exists a constant $\varepsilon > 0$ such that finding an $\varepsilon$-approximate consensus halving for $n$ agents with monotonic utilities that uses at most $n$ cuts is PPAD-hard, even if all agents have symmetric-threshold utilities.
\end{theorem}

\begin{proof}
We prove this result by reducing from a modified version of the \emph{generalized circuit} problem. The generalized circuit problem is the main tool that has been used (implicitly or explicitly) to prove hardness of computing Nash equilibria in various settings \citep{ChenDeTe09,DaskalakisGoPa09,Rubinstein18}. A generalized circuit is a generalization of an arithmetic circuit, because it allows \emph{cycles}, which means that instead of a simple computation, the circuit now represents a constraint satisfaction problem. The version of the problem we use is different from the standard one in two aspects. First, instead of the domain $[0,1]$, we use $[-1,1]$, which is more adapted to the consensus halving problem. Second, we will only allow the circuit to use three types of arithmetic gates. As we will show below, these modifications do not change the complexity of the problem.

Formally, we consider the following simplified generalized circuits.

\begin{definition}
A \emph{simple generalized circuit} is a pair $(V,\mathcal{T})$, where $V$ is a set of nodes and $\mathcal{T}$ is a set of gates. Every gate $T \in \mathcal{T}$ is a $5$-tuple $T=(G,u_1,u_2,v,\zeta)$ where $G \in \{G_+,G_{\times-|\zeta|},G_1\}$ is the type of gate, $u_1,u_2$ are the input nodes (if applicable), $\zeta \in (0,1]$ is the parameter (if applicable), and $v$ is the output node. In more detail:
\begin{itemize}
	\item if $G=G_+$, then $u_1, u_2, v \in V$ (distinct) and $\zeta = nil$,
	\item if $G=G_{\times-|\zeta|}$, then $u_1,v \in V$ (distinct), $u_2 = nil$ and $\zeta \in (0,1]$,
	\item if $G=G_1$, then $u_1 = u_2 = \zeta = nil$ and $v \in V$.
\end{itemize}
We require that for any two gates $T=(G,u_1,u_2,v,\zeta)$ and $T'=(G',u_1',u_2',v',\zeta')$ in $\mathcal{T}$ with $T \neq T'$, it holds that $v \neq v'$.
\end{definition}

Before we proceed, let us introduce some notation. We let $T_{[-1,1]} : \mathbb{R} \to [-1,1]$ denote \emph{truncation} to $[-1,1]$, i.e., $T_{[-1,1]}(x) = \max\{-1, \min\{1, x\} \}$. Similarly, we also let $T_{[0,1]}$ denote truncation to $[0,1]$. Finally, we use the notation $x = y \pm z$ as a shorthand for $|x-y| \leq z$.

\begin{definition}
Let $\varepsilon > 0$. The problem \textsc{$\varepsilon$-simple-Gcircuit} is defined as follows: given a simple generalized circuit $(V,\mathcal{T})$, find an assignment $\xx : V \to [-1,1]$ that $\varepsilon$-approximately satisfies all the gates $T=(G,u_1,u_2,v,\zeta)$ in $\mathcal{T}$, namely:
\begin{itemize}
	\item if $G=G_+$, then $\xx[v] = T_{[-1,1]}(\xx[u_1] + \xx[u_2]) \pm \varepsilon$, \hfill (\emph{addition})
	\item if $G=G_{\times-|\zeta|}$, then $\xx[v] = - |\zeta| \cdot \xx[u_1] \pm \varepsilon$, \hfill (\emph{multiplication by $-|\zeta|$ for $\zeta \in (0,1]$})
	\item if $G=G_1$, then $\xx[v] = 1 \pm \varepsilon$. \hfill (\emph{constant $1$})
\end{itemize}
\end{definition}

As mentioned earlier, it turns out that this modified version of the generalized circuit problem is also PPAD-hard. This can be proved by reducing from the standard \textsc{$\varepsilon$-Gcircuit} problem, which was shown to be PPAD-hard even for constant $\varepsilon$ by \citet{Rubinstein18}. The idea is that these simple gates are enough to simulate all the gates in the standard version of the problem. Both problems are in fact PPAD-complete, since they can be reduced to the problem of finding an approximate Brouwer fixed point, but here we are only interested in the hardness. 

\begin{lemma}\label{lem:simple-g-circuit}
There exists a constant $\varepsilon > 0$ such that the \textsc{$\varepsilon$-simple-Gcircuit} problem is PPAD-hard.
\end{lemma}

The proof of Lemma~\ref{lem:simple-g-circuit} can be found in Appendix~\ref{app:proof-g-circuit}.

Let $\widehat{\varepsilon} > 0$ be a constant for which the \textsc{$\widehat{\varepsilon}$-simple-Gcircuit} problem is PPAD-hard. We will now show that the \textsc{$\widehat{\varepsilon}$-simple-Gcircuit} problem reduces to the problem of finding an $\varepsilon$-approximate consensus halving for $n$ agents with symmetric-threshold utilities that uses at most $n$ cuts.

Let $(V,\mathcal{T})$ be an instance of \textsc{$\widehat{\varepsilon}$-simple-Gcircuit}. Partition $V$ into four sets $V_0 \cup V_+ \cup V_\times \cup V_1$, where
\begin{itemize}
	\item $V_0$ contains every node that is not the output of any gate in $\mathcal{T}$,
	\item $V_+$ contains every node that is the output of a $G_+$ gate in $\mathcal{T}$,
	\item $V_\times$ contains every node that is the output of a $G_{\times-|\zeta|}$ gate in $\mathcal{T}$,
	\item $V_1$ contains every node that is the output of a $G_1$ gate in $\mathcal{T}$.
\end{itemize}

We construct a consensus halving instance with $n=2|V_+| + |V_\times| + |V_1|$ agents and $m=|V_0| + 2|V_+| + |V_\times| + |V_1| + 1$ items. For any node $v \in V_+ \cup V_\times \cup V_1$, let $i(v) \in N = [n]$ denote the corresponding agent, and for every $v \in V_+$, let $i'(v) \in N$ denote the second corresponding agent. For every $v \in V$, let $j(v) \in M = [m]$ denote the corresponding item, and for every $v \in V_+$, let $j'(v) \in M$ denote the second corresponding item. Finally, let $j^* \in M$ denote the single remaining item, which we call the \emph{special item}.

It remains to specify the utility functions for the agents and the constant $\varepsilon > 0$. We will see below that in any partition of $M$ into two fractional sets of items $(M_1, M_2)$, there is a simple way to associate a value $\text{val}(j) \in [-1,1]$ to every item $j \in M$. We will pick the agents' utilities so that in any $\varepsilon$-approximate consensus halving (with at most $n$ cuts), these values must satisfy the gate constraints in~$\mathcal{T}$.

\paragraph{Value Encoding.}
Consider any partition of $M$ into two fractional sets of items $(M_1, M_2)$. Let $x_j \in [0,1]$ denote the fraction of item $j$ in $M_1$. This fraction $x_j \in [0,1]$ encodes a number $\text{val}(j) \in [-1,1]$ as follows:
\begin{equation*}
\text{val}(j) = \left\{\begin{tabular}{cl}
$-1$ & if $x_j \leq 1/3$;\\
$6(x_j-1/2)$ & if $1/3 < x_j < 2/3$;\\
$1$ & if $x_j \geq 2/3$.
\end{tabular}\right.
\end{equation*}
In other words, $\text{val}(j) = T_{[-1,1]}(6x_j-3)$.

The main idea of the reduction is that the value $\xx[v]$ of node $v \in V$ will be given by $\text{val}(j(v))$. Next, we show how to pick the utility functions in order to enforce the gate constraints in $\mathcal{T}$. In the construction below we assume that $\varepsilon \leq 1/10$; the exact value of $\varepsilon$ will be picked at the end.

\paragraph{$G_{\times -|\zeta|}$ gates.} For any gate $(G_{\times -|\zeta|}, u_1, nil, v, \zeta) \in \mathcal{T}$, where $u_1 \in V \setminus \{v\}$, $v \in V_\times$ and $\zeta \in (0,1]$, we do the following. Let $j_1=j(u_1)$, $j=j(v)$ and $i=i(v)$. We want to ensure that in any solution to $\varepsilon$-approximate consensus halving, we have $\text{val}(j) = -\zeta \cdot \text{val}(j_1) \pm \widehat{\varepsilon}$. To achieve this we define the symmetric-threshold utility function of agent $i$ as follows. For any item $\ell \notin \{j_1,j\}$, we let $u_i(\ell)=0$ and $c_{i\ell}=0$. We let $u_i(j)=1/\zeta$ and $c_{ij}=0$. For $j_1$ we use what we call a \emph{standard input utility function}, which is defined as follows: $u_i(j_1)=1/3$ and $c_{ij_1}=1/3$. Note that $u_i(M)=1/3+1/\zeta$.

Consider any $\varepsilon$-approximate consensus halving $(M_1,M_2)$. Then, it must hold that $u_i(M_1) = u_i(M_2) \pm \varepsilon \cdot u_i(M)$. First of all, since $u_i(j) > u_i(M \setminus \{j\}) + \varepsilon \cdot u_i(M)$ and by monotonicity, this implies that item $j$ must be fractional in the partition $(M_1,M_2)$, i.e., $x_j \in (0,1)$. Furthermore, we must have
$$f_{ij_1}(x_{j_1}) u_i(j_1) + f_{ij}(x_j) u_i(j) = f_{ij_1}(1-x_{j_1}) u_i(j_1) + f_{ij}(1-x_j) u_i(j) \pm \varepsilon \cdot u_i(M).$$
Since $f_{i\ell}(1-x)=1-f_{i\ell}(x)$ for any $x \in [0,1]$ and $\ell \in M$, this equation can be rewritten as
$$(2f_{ij}(x_j) - 1) u_i(j) = -(2f_{ij_1}(x_{j_1})-1) u_i(j_1) \pm \varepsilon \cdot u_i(M).$$
By noting that $f_{i\ell}(x) = T_{[0,1]}((x-c_{i\ell})/(1-2c_{i\ell}))$, we obtain
$$(2x_j - 1) \cdot (1/\zeta) = -(2T_{[0,1]}(3x_{j_1}-1)-1) \cdot (1/3) \pm \varepsilon \cdot u_i(M).$$
Finally, by observing that $2T_{[0,1]}(3x_{j_1}-1)-1 = T_{[-1,1]}(2(3x_{j_1}-1)-1) = T_{[-1,1]}(6x_{j_1} - 3) = \text{val}(j_1)$, we obtain
$$(6x_j-3) = - \zeta \cdot \text{val}(j_1) \pm 3 \zeta \varepsilon \cdot u_i(M).$$
Now, this yields
$$\text{val}(j) = T_{[-1,1]}(6x_j-3) = T_{[-1,1]}(- \zeta \cdot \text{val}(j_1)) \pm 3 \zeta \varepsilon \cdot u_i(M) = - \zeta \cdot \text{val}(j_1) \pm 4\varepsilon$$
where we used the fact that $- \zeta \cdot \text{val}(j_1) \in [-1,1]$, $u_i(M)=1/3+1/\zeta$ and $\zeta \leq 1$. Thus, as long as $4\varepsilon \leq \widehat{\varepsilon}$, this construction correctly enforces the gate constraint.

\paragraph{$G_1$ gates.} For any gate $(G_1, nil, nil, v, nil) \in \mathcal{T}$, where $v \in V_1$, we do the following. Let $j=j(v)$ and $i=i(v)$. We use the same construction as for $G_{\times -|\zeta|}$ gates with $j_1=j^*$ (the special item) and $\zeta = 1$. By the same arguments, it follows that in any $\varepsilon$-approximate solution it must hold that $\text{val}(j) = - \text{val}(j^*) \pm 4\varepsilon$ and item $j$ must be fractional, i.e., $x_j \in (0,1)$. Thus, as long as $4\varepsilon \leq \widehat{\varepsilon}$ and $\text{val}(j^*)=-1$, this correctly enforces the gate constraint.

\paragraph{$G_+$ gates.} For any gate $(G_+, u_1, u_2, v, nil) \in \mathcal{T}$, where $u_1 \in V \setminus \{v\}$, $u_2 \in V \setminus \{v,u_1\}$ and $v \in V_+$, we do the following. Let $j_1=j(u_1)$, $j_2=j(u_2)$, $j=j(v)$ and $j'=j'(v)$. We are going to ensure that $\text{val}(j') = -T_{[-1,1]}(\text{val}(j_1) + \text{val}(j_2)) \pm \widehat{\varepsilon}/2$ and $\text{val}(j) = -\text{val}(j') \pm \widehat{\varepsilon}/2$. Together, these two constraints will enforce the gate constraint. The second constraint can easily be enforced by using the same construction as for $G_{\times -|\zeta|}$ with $j_1=j'(v)$, $j=j(v)$, $i=i(v)$ and $\zeta = 1$. By the same arguments, this will yield an error of at most $\widehat{\varepsilon}/2$, as long as $8\varepsilon \leq \widehat{\varepsilon}$, and ensure that item $j$ is fractional.

To enforce the first constraint, we define the utilities of agent $i'=i'(v)$ as follows. For any item $\ell \notin \{j_1,j_2,j'\}$, we let $u_{i'}(\ell)=0$ and $c_{i'\ell}=0$. We let $u_{i'}(j')=1$ and $c_{i'j'}=0$. For $j_1$ and $j_2$ we use the standard input utility function as defined earlier. Note that $u_{i'}(M)=5/3$.

Consider any $\varepsilon$-approximate consensus halving $(M_1,M_2)$. Then, it must hold that $u_{i'}(M_1) = u_{i'}(M_2) \pm \varepsilon \cdot u_{i'}(M)$. First of all, since $u_{i'}(j') > u_{i'}(M \setminus \{j'\}) + \varepsilon \cdot u_{i'}(M)$ and by monotonicity, this implies that item $j'$ must be fractional in the partition $(M_1,M_2)$, i.e., $x_{j'} \in (0,1)$. Furthermore, by the same arguments as for $G_{\times -|\zeta|}$ gates, we obtain that
$$6x_{j'}-3 = - \text{val}(j_1) - \text{val}(j_2) \pm 3 \varepsilon \cdot u_i(M).$$
Since $\text{val}(j') = T_{[-1,1]}(6x_{j'}-3)$, it follows that $\text{val}(j') = - T_{[-1,1]}(\text{val}(j_1) + \text{val}(j_2)) \pm 5\varepsilon$. Thus, this constraint is correctly enforced as long as $10\varepsilon \leq \widehat{\varepsilon}$.

\bigskip

We are now ready to complete the proof. Set $\varepsilon = \widehat{\varepsilon}/10$. Consider any $\varepsilon$-approximate consensus halving $(M_1,M_2)$ that uses at most $n$ cuts. We claim that letting $\xx[v] = \text{val}(j(v))$ for all $v \in V$ yields a solution to the \textsc{$\widehat{\varepsilon}$-simple-Gcircuit} instance. Indeed, by construction, all gates of type $G_+$ and $G_{\times -|\zeta|}$ are correctly enforced. For gates of type $G_1$, they will be correctly enforced if $\text{val}(j^*)=-1$, which we now prove. Note that in our construction, we have ensured that for every $v \in V_+ \cup V_\times \cup V_1$, item $j(v)$ must be fractional, and for every $v \in V_+$, item $j'(v)$ must also be fractional. Since these $2|V_+| + |V_\times| + |V_1| = n$ items are fractional, and we used at most $n$ cuts, this means that all other items are \emph{not} fractional. In particular, $j^*$ is not fractional, i.e., $x_{j^*} \in \{0,1\}$. Without loss of generality, assume that $x_{j^*} = 0$ (if $x_{j^*} = 1$, then swap the roles of $M_1$ and $M_2$). It follows that $\text{val}(j^*)=-1$.
\end{proof}

\section{Connections to Agreeable Sets}
\label{sec:agreeable}

We now present some implications of results from consensus halving on the setting of computing agreeable sets.
Let us first formally define the agreeable set problem, introduced by \citet{ManurangsiSu19}.\footnote{The notion of agreeability was introduced in an earlier conference version of the paper \citep{Suksompong16}. \citet{Gourves19} considered an extension of the problem that takes into account matroidal constraints.}
As in consensus halving, there is a set $N$ of $n$ agents and a set $M$ of $m$ items.
Agent $i$ has a monotonic utility function $u_i$ over \emph{non-fractional} sets of items, where we assume the normalization $u_i(\emptyset) = 0$; this corresponds to a set function.

\begin{definition}
\label{def:agreeable}
A subset of items $M'\subseteq M$ is said to be \emph{agreeable to agent~$i$} if $u_i(M')\geq u_i(M\backslash M')$.
\end{definition}

As one of their main results, \citet{ManurangsiSu19} showed that for any $n$ and $m$, there exists a set of at most $\min\left\{\lfloor\frac{m+n}{2}\rfloor, m\right\}$ items that is agreeable to all agents, and this bound is tight.
Their proof relies on a graph-theoretic statement often referred to as ``Kneser's conjecture'', which specifies the chromatic number for a particular class of graphs called Kneser graphs.
Here we present a short alternative proof that works by arranging the items on a line in arbitrary order, applying consensus halving, and rounding the resulting fractional partition.
As a bonus, our proof yields an agreeable set that is composed of at most $\lfloor n/2\rfloor + 1$ blocks on the line.

\begin{theorem}[\citet{ManurangsiSu19}]
\label{thm:agreeable-existence}
For $n$ agents with monotonic utilities, there exists a subset $M'\subseteq M$ such that
\[
|M'| \leq \min\left\{\left\lfloor\frac{m+n}{2}\right\rfloor, m\right\}
\]
and $M'$ is agreeable to all agents.
\end{theorem}

\begin{proof}
Let $s=\left\lfloor\frac{m+n}{2}\right\rfloor$.
If $s\geq m$, the entire set of items $M$ has size $m=\min\{s,m\}$ and is agreeable to all agents due to monotonicity, so we may assume that $s\leq m$.
Arrange the items on a line in arbitrary order, and extend the utility functions of the agents to fractional sets of items in a continuous and monotonic fashion.\footnote{For example, one can use the \emph{Lov\'{a}sz extension} or the \emph{multilinear extension} (see Section~\ref{app-subsec:continuous-extensions}). \label{fn:monotonic-extension}}
Consider a consensus halving with respect to the extended utilities that uses at most $n$ cuts on the line; some of the cuts may cut through items, whereas the remaining cuts are between adjacent items.
Let $r\leq n$ be the number of items that are cut by at least one cut.
Without loss of generality, assume that the first part $M'$ contains no more full items than the second part $M''$, so $M'$ contains at most $\left\lfloor\frac{m-r}{2}\right\rfloor$ full items.
By moving all cut items from $M''$ to $M'$ in their entirety, $M'$ contains at most $\left\lfloor\frac{m-r}{2}\right\rfloor + r = \left\lfloor\frac{m+r}{2}\right\rfloor \leq s$ items.
Since we start with a consensus halving and only move fractional items from $M''$ to $M'$, we have that $M'$ is agreeable to all agents.
Moreover, one can check that $M'$ is composed of at most $\left\lceil\frac{n+1}{2}\right\rceil = \left\lfloor\frac{n}{2}\right\rfloor + 1$ blocks on the line.
\end{proof}

In light of Theorem~\ref{thm:agreeable-existence}, an important question is how efficiently we can compute an agreeable set whose size matches the worst-case bound.
\citet{ManurangsiSu19} addressed this question by providing a polynomial-time algorithm for two agents with monotonic utilities and three agents with ``responsive'' utilities, a class that lies between additive and monotonic utilities.
They left the complexity for higher numbers of agents as an open question, and conjectured that the problem is hard even when the number of agents is a larger constant.
We show that this is in fact not the case: the problem can be solved efficiently for any number of agents with additive utilities, as well as for any \emph{constant} number of agents with monotonic utilities.
Note that since the input of the problem for monotonic utilities can involve an exponential number of values (even for constant $n$), and consequently may not admit a succinct representation, we assume a ``utility oracle model'' in which the algorithm is allowed to query the utility $u_i(M')$ for any $i\in N$ and $M'\subseteq M$.

\begin{theorem}
\label{thm:agreeable-algo}
There exists a polynomial-time algorithm that computes a set containing at most $\min\left\{\left\lfloor\frac{m+n}{2}\right\rfloor, m\right\}$ items that is agreeable to all agents, for each of the following two cases:
\begin{enumerate}[label=(\roman*)]
\item All agents have additive utilities.
\item All agents have monotonic utilities and the number of agents is constant (assuming access to a utility oracle).
\end{enumerate}
\end{theorem}

\begin{proof}
Similarly to Theorem~\ref{thm:agreeable-existence}, if $n\geq m$ we can simply include all items in our set, so we may focus on the case $n\leq m$.
For (i), we first use our polynomial-time algorithm from Theorem~\ref{thm:polytope} to find a consensus halving, and then compute an agreeable set of size at most $\left\lfloor\frac{m+n}{2}\right\rfloor$ by rounding the consensus halving as in the proof of Theorem~\ref{thm:agreeable-existence}.

Next, consider (ii). 
Recall that for any ordering of the items on a line, Theorem~\ref{thm:agreeable-existence} guarantees the existence of an agreeable set of size at most $\left\lfloor\frac{m+n}{2}\right\rfloor$ involving no more than $n$ cuts on the line.
Fix an ordering of the items; we will perform a brute-force search over all (non-fractional) partitions involving at most $n$ cuts with respect to the ordering.
For $t\in[n]$, there are $O(m^t)$ ways to place $t$ cuts, and for each way, we have two candidate sets to check: one including the leftmost item, and one not including it.
A candidate set is valid if and only if it has size at most $\left\lfloor\frac{m+n}{2}\right\rfloor$ and is agreeable to all agents.
Hence the brute-force search runs in time $\sum_{t=1}^n O(m^t) = O(n\cdot m^n) = O(m^n)$, which is polynomial since $n$ is constant.
\end{proof}

\section{Consensus $k$-Splitting}
\label{sec:k-splitting}

In this section, we address two important generalizations of consensus halving, both of which were mentioned by \citet{SimmonsSu03}.
In \emph{consensus splitting}, instead of dividing the items into two equal parts, we want to divide them into two parts so that all agents agree that the split satisfies some given ratio, say two-to-one.
In \emph{consensus $1/k$-division}, we want to divide the items into $k$ parts that all agents agree are equal.
We consider a problem that generalizes both of these problems at once.
\begin{definition}
\label{def:consensus-k-splitting}
Let $\alpha_1,\dots,\alpha_k > 0$ be real numbers such that $\alpha_1+\dots+\alpha_k = 1$.
A \emph{consensus $k$-splitting with ratios $\alpha_1,\dots,\alpha_k$} is a partition of $M$ into $k$ fractional sets of items $M_1,\dots,M_k$ such that
\begin{align*}
\frac{u_i(M_1)}{\alpha_1} = \frac{u_i(M_2)}{\alpha_2} = \dots = \frac{u_i(M_k)}{\alpha_k} & & \forall i \in N.
\end{align*}
When the ratios are clear from context, we will simply refer to such a partition as a \emph{consensus $k$-splitting}.
\end{definition}
\noindent As in Section~\ref{sec:additive}, we will assume that the utility functions are additive, in which case our desired condition is equivalent to $u_i(M_\ell) = \alpha_\ell \cdot u_i(M)$ for all $i\in N$ and $\ell\in[k]$.

While there is no reason to cut an item more than once in consensus halving, one may sometimes wish to cut the same item multiple times in consensus $k$-splitting in order to split the item across three or more parts.
Hence, even though the number of cuts made is always at least the number of items cut, the two quantities are not necessarily the same in consensus $k$-splitting.
If there are $n$ items and each agent only values a single distinct item, then it is clear that we already need to make $(k-1)n$ cuts for any ratios $\alpha_1,\dots,\alpha_k$, in particular $k-1$ cuts for each item.
Nevertheless, it could still be that for some ratios, it is always possible to achieve a consensus $k$-splitting by cutting fewer than $(k-1)n$ items.
We show that this is not the case: for any set of ratios, cutting $(k-1)n$ items is necessary in the worst case.

\begin{theorem}
\label{thm:k-splitting-worstcase}
For any ratios $\alpha_1,\dots,\alpha_k > 0$, there exists an instance with additive utilities in which any consensus $k$-splitting with these ratios cuts at least $(k-1)n$ items.
\end{theorem}

\begin{proof}
Fix $\alpha_1,\dots,\alpha_k > 0$.
We construct an instance such that each agent $i$ has utility $1/b$ for each of the $b$ items in a set $B_i$, where $b$ is an integer that we will choose later, and utility $0$ for every other item. The sets $B_1,\dots,B_n$ are pairwise disjoint.
Note that $u_i(M) = u_i(B_i) = 1$ for every $i$.
It suffices to choose $b$ such that at least $k-1$ items in each set $B_i$ must be cut in any consensus $k$-splitting with ratios $\alpha_1,\dots,\alpha_k$.
By symmetry, we may focus on the first agent and the corresponding set $B_1$.

For any real number $x$, denote by $\lfloor x\rfloor$ its floor function, and let $\{x\} = x-\lfloor x\rfloor$.
We will choose $b$ such that
\begin{equation}
\label{eq:sum-of-fractionals}
\{\alpha_1b\}+\{\alpha_2b\}+\dots+\{\alpha_kb\} > k-2.
\end{equation}
To see why this is sufficient, observe that each \emph{uncut} item must belong to one of the $k$ parts in its entirety.
The number of uncut items in $B_1$ is therefore at most
\[
\left\lfloor\frac{\alpha_1}{1/b}\right\rfloor + \dots + \left\lfloor\frac{\alpha_k}{1/b}\right\rfloor
=
\lfloor\alpha_1b\rfloor + \dots + \lfloor\alpha_kb\rfloor,
\]
meaning that the number of cut items in $B_1$ is at least
\begin{align*}
b - (\lfloor\alpha_1b\rfloor + \dots + \lfloor\alpha_kb\rfloor)
&= (\alpha_1b + \dots + \alpha_kb) - (\lfloor\alpha_1b\rfloor + \dots + \lfloor\alpha_kb\rfloor) \\
&= (\alpha_1b - \lfloor\alpha_1b\rfloor) + \dots + (\alpha_kb - \lfloor\alpha_kb\rfloor) \\
&= \{\alpha_1b\} + \dots + \{\alpha_kb\} \\
&> k-2,
\end{align*}
where the first equality follows from $\alpha_1 + \dots + \alpha_k = 1$.
Since $b,\lfloor\alpha_1b\rfloor,\dots,\lfloor\alpha_kb\rfloor$ are all integers, this implies that at least $k-1$ items in $B_1$ must be cut.

It remains to show the existence of $b$ for which \eqref{eq:sum-of-fractionals} is satisfied.
Let $s$ be an integer such that 
\[
s > \max\left\{k,\frac{1}{\alpha_1},\dots,\frac{1}{\alpha_k},\frac{1}{1-\alpha_1},\dots,\frac{1}{1-\alpha_k}\right\}.
\]
Divide the interval $[0,1]$ into subintervals of length at most $1/s$ each.
By the pigeonhole principle, there exist positive integers $p,q$ such that $q\geq p+2$, and $\{\alpha_ip\}$ and $\{\alpha_iq\}$ fall in the same subinterval for every $i\in[k]$.
Letting $c = q-p$, we have that for each $i\in[k]$, either $\{\alpha_ic\} < 1/s$ or $\{\alpha_ic\} > 1-1/s$.

Take $b = c-1 \geq 1$.
From our choice of $s$, we have $1/s < \alpha_i < 1- 1/s$ for all $i\in [k]$.
Thus, for each $i$, if $\{\alpha_ic\} < 1/s$ then $\{\alpha_ic\} < \alpha_i$, while if $\{\alpha_ic\} > 1-1/s$ then $\{\alpha_ic\} > \alpha_i$.
In either case, we have $\{\alpha_ib\} = \{\alpha_ic - \alpha_i\} > 1 - 1/s - \alpha_i$, so
\[
\{\alpha_1b\}+\dots+\{\alpha_kb\}
> k - k/s - (\alpha_1 + \dots + \alpha_k)
> k-2,
\]
where we use the assumption that $s>k$.
Hence \eqref{eq:sum-of-fractionals} is satisfied, and the proof is complete.
\end{proof}

Next, we show that computing a consensus $k$-splitting with at most $(k-1)n$ cuts can be done efficiently using a generalization of our algorithm for consensus halving (Theorem~\ref{thm:polytope}).
Note that such a splitting also cuts at most $(k-1)n$ items.

\begin{theorem}
\label{thm:polytope-k-splitting}
For $n$ agents with additive utilities and ratios $\alpha_1,\dots,\alpha_k$, there is a polynomial-time algorithm that computes a consensus $k$-splitting with these ratios using at most $(k-1)\cdot\min\{n,m\}$ cuts.
\end{theorem}

\begin{proof}
Let us start with the case $k = 2$, which can then be used as a subroutine for the case $k > 2$. Our algorithm for consensus 2-splitting generalizes the consensus halving algorithm in Theorem~\ref{thm:polytope}, so we only highlight the differences. To find a consensus 2-splitting with ratios $\alpha_1, \alpha_2$, the only change to the algorithm in Theorem~\ref{thm:polytope} is that we initialize $x_1 = \cdots = x_m = \alpha_1$ and let $S$ be the set of $n$ equations $\sum_{j \in M} (y_j - \alpha_1) \cdot u_i(j)$ for $i\in N$. By analogous arguments as in Theorem~\ref{thm:polytope}, this modified algorithm produces a consensus 2-splitting with ratios $\alpha_1, \alpha_2$ in polynomial time and uses at most $\min\{n, m\}$ cuts.

We now move on to the case $k > 2$. In this case, we simply apply the above consensus 2-splitting algorithm successively, each time producing one additional part at the expense of at most $\min\{n, m\}$ cuts. This is stated more precisely below.
\begin{enumerate}
\item Let $M_{\text{remaining}} = M$.
\item For $\ell = 1, \dots, k - 1$:
\begin{enumerate}
\item $(M_\ell, M_{\text{remaining}}) \leftarrow$ consensus 2-splitting of $M_{\text{remaining}}$ with ratios $\frac{\alpha_\ell}{\alpha_\ell + \dots + \alpha_n}, \frac{\alpha_{\ell + 1} + \dots + \alpha_n}{\alpha_\ell + \dots + \alpha_n}$.
\end{enumerate}
\item Output $(M_1, \dots, M_{k - 1}, M_{\text{remaining}})$
\end{enumerate}
It is clear that the output is a consensus $k$-splitting with ratios $\alpha_1, \dots, \alpha_k$, and that the algorithm runs in polynomial time. Finally, observe that each time we apply the consensus 2-splitting algorithm, if there are $m'$ items left, we additionally use at most $\min\{n, m'\} \leq \min\{n,m\}$ cuts. As a result, the total number of cuts is at most $(k - 1) \cdot \min\{n, m\}$, as desired.
\end{proof}

As in Theorem~\ref{thm:polytope}, our algorithm does not require the nonnegativity assumption on the utilities and therefore works for combinations of goods and chores.

When the items lie on a line, there is always a consensus halving that makes at most $n$ cuts on the line and therefore cuts at most $n$ items---this matches the upper bound on the number of items cut in the absence of a linear order.
Theorem~\ref{thm:polytope-k-splitting} shows that the bound $n$ continues to hold for consensus splitting into two parts with any ratios.
As we show next, however, this bound is no longer achievable for some ratios with ordered items, thereby demonstrating another difference that the lack of linear order makes.\footnote{See the definition of the consensus halving problem on a line before Theorem~\ref{thm:PPA-line}.}

\begin{theorem}
\label{thm:k-splitting-more-cuts}
Let $n\geq 2$, $k=2$ and $(\alpha_1,\alpha_2) = (\frac{1}{n},\frac{n-1}{n})$.
There exists an instance such that the $n$ agents have additive utilities, the items lie on a line, and any consensus $k$-splitting with ratios $\alpha_1$ and $\alpha_2$ makes at least $2n-4$ cuts on the line.
\end{theorem}

\begin{proof}
We discretize a slight modification of an instance used by \citet{StromquistWo85} to show a lower bound on the number of cuts when the resource is represented by a one-dimensional circle.
Suppose that there are $n^2-1$ ``primary items'', which we label as $1,2,\dots,n^2-1$ according to their linear order.
Moreover, there are $n^2-2$ ``secondary items'', one between every adjacent pair of primary items.
The utilities of the agents are as follows:
\begin{itemize}
\item For $i\in[n-1]$, agent $i$ has utility $\frac{1}{n+1}$ for each of the $n+1$ primary items $i, i+(n-1), \dots, i+n(n-1)$, and utility $0$ for all secondary items.
\item Agent $n$ has value $\frac{1}{n^2-2}$ for each secondary item, and value $0$ for all primary items.
\end{itemize}

Note that $u_i(M)=1$ for all $i$.
Let $M'$ be a fractional set of items for which all agents have utility $1/n$.
Since each agent $i\in[n-1]$ has utility $\frac{1}{n+1}$ for a primary item, $M'$ must contain a positive fraction of at least two primary items that the agent values.
These items are disjoint for different agents, so $M'$ necessarily contains a positive fraction of at least $2n-2$ primary items.
On the other hand, the utility function of agent $n$ implies that $M'$ can contain at most $\left\lfloor\frac{1/n}{1/(n^2-2)}\right\rfloor = n-1$ entire secondary items.

Suppose that $M'$ is composed of $r$ non-adjacent intervals $I_1,\dots,I_r$.
Notice that for any interval $I$ on the line, if the interval contains a positive fraction of $t_1(I)$ primary items, along with $t_2(I)$ \emph{entire} secondary items, then $t_1(I)\leq t_2(I) + 1$.
Hence, we have
\[
2n-2 \leq
\sum_{i=1}^r t_1(I_i) \leq \sum_{i=1}^r t_2(I_i) + r \leq
n-1 + r,
\]
implying that $r\geq n-1$.
This means that the consensus $2$-splitting with $M'$ as one part involves at least $2(n-1) = 2n-2$ cuts, possibly including endpoints of the line.
At most two of these cuts can correspond to endpoints, implying that the number of cuts made is at least $2n-4$, as desired.
\end{proof}

For consensus halving, Theorem~\ref{thm:asymptotic-n-cuts} shows that in a random instance, any solution almost surely uses at least the worst-case number of cuts $\min\{n,m\}$.
One might consequently expect that an analogous statement holds for consensus $k$-splitting, with $(k-1)\cdot\min\{n,m\}$ cuts almost always being required. 
However, we show that this is not true: even in the simple case where $n=1$ and the agent's utilities are drawn from the uniform distribution over $[0,1]$, it is likely that we only need to make one cut (instead of $k-1$) for large $m$.

\begin{theorem} \label{thm:asymptotic-upper-bound}
Let $n=1$, and suppose that the agent's utility for each item is drawn independently from the uniform distribution on $[0,1]$.
For any ratios $\alpha_1,\dots,\alpha_k > 0$, with probability approaching $1$ as $m\rightarrow\infty$, there exists a consensus $k$-splitting with these ratios using at most one cut.
Moreover, there is a polynomial-time algorithm that computes such a solution.
\end{theorem}

In what follows, we denote the agent's utility function by $u$, and say that an event happens ``with high probability'' if the probability that it happens approaches $1$ as $m\rightarrow\infty$.
The proof of Theorem~\ref{thm:asymptotic-upper-bound} proceeds by identifying a simple (deterministic) condition that guarantees a solution cutting only a single item; this is done in Lemma~\ref{lem:deterministic-condition-upper-bound}. Then, we show that this condition is satisfied with high probability.

\begin{lemma} \label{lem:deterministic-condition-upper-bound}
Suppose that there is a single agent. 
Let $j^* := \argmax_j u(j)$ denote a most-preferred item, and let $M_{\text{low-utility}} := \{j \in M \mid u(j) \leq \frac{1}{k} \cdot u(j^*)\}$ denote the set of items whose utility is less than $1/k$ times the utility of $j^*$. 
For any ratios $\alpha_1,\dots,\alpha_k > 0$, if $\sum_{j \in M_{\text{low-utility}}} u(j) \geq k \cdot u(j^*)$, then there is a consensus $k$-splitting with these ratios that cuts only $j^*$.
Moreover, there is a polynomial-time algorithm that computes such a solution.
\end{lemma}

\begin{proof}
For each $\ell \in [k]$, let $w_\ell := \alpha_\ell \cdot (\sum_{j \in M} u(j))$ be the ``target utility'' for part $\ell$ of the partition.
Consider the following greedy algorithm. 
\begin{itemize}
\item Let $P_1 = \cdots = P_k = \emptyset$.
\item Let $M^0 = M \setminus \{j^*\}$ and $j^{\text{max}} = \argmax_{j \in M^0} u(j)$.
\item While there exists $\ell \in [k]$ such that $u(P_{\ell} \cup \{j^{\text{max}}\}) \leq w_\ell$:
\begin{itemize}
\item Add $j^{\text{max}}$ to $P_{\ell}$.
\item Remove $j^{\text{max}}$ from $M^0$.
If $M^0 = \emptyset$, terminate.
Else, update $j^{\text{max}} = \argmax_{j \in M^0} u(j)$.
\end{itemize}
\end{itemize}
The algorithm clearly runs in polynomial time.
We claim that it terminates with $M^0 = \emptyset$ provided that $\sum_{j \in M_{\text{low-utility}}} u(j) \geq k \cdot u(j^*)$. 
This implies the statement of the lemma, because it would then suffice to split only item $j^*$.

Suppose for the sake of contradiction that $M^0 \ne \emptyset$ at the end of the execution. 
Consider the following two cases, based on whether $j^{\text{max}}$ at termination belongs to $M_{\text{low-utility}}$.
\begin{itemize}
\item \emph{Case 1:} $j^{\text{max}} \notin M_{\text{low-utility}}$. 
Since the algorithm terminates, it must be that $u(P_\ell) > w_\ell - u(j^{\text{max}}) \geq w_\ell - u(j^*)$ for each $\ell$. 
Summing this over $\ell \in [k]$, we get
\begin{align*}
u\left(P_1 \cup \dots \cup P_k\right) > w_1 + \dots + w_k - k \cdot u(j^*) = u(M) - k \cdot u(j^*).
\end{align*}
On the other hand, since $j^{\text{max}} \notin M_{\text{low-utility}}$, it must be that $M_{\text{low-utility}}$ is disjoint from $P_1 \cup \dots \cup P_k$. As a result, we have
\begin{align*}
u\left(P_1 \cup \dots \cup P_k\right) \leq u(M) - \sum_{j \in M_{\text{low-utility}}} u(j)
\leq u(M) - k \cdot u(j^*),
\end{align*}
where the second inequality is from the assumption of the lemma.
The above two inequalities imply the desired contradiction.
\item \emph{Case 2:} $j^{\text{max}} \in M_{\text{low-utility}}$. 
In this case, we must have $u(P_\ell) > w_\ell - u(j^{\text{max}}) \geq w_\ell - u(j^*) / k$ for each $\ell$. 
Summing this over $\ell \in [k]$, we get
\begin{align*}
u\left(P_1 \cup \cdots \cup P_k\right) > w_1 + \cdots + w_k - u(j^*) = u(M) - u(j^*).
\end{align*}
However, since $j^* \notin P_1 \cup \dots \cup P_k$, we have $u\left(P_1 \cup \dots \cup P_k\right) \leq u(M) - u(j^*)$, which is a contradiction.
\end{itemize}
In both cases, we arrive at a contradiction, and our proof is complete.
\end{proof}

With Lemma~\ref{lem:deterministic-condition-upper-bound} ready, we can now prove Theorem~\ref{thm:asymptotic-upper-bound}.

\begin{proof}[Proof of Theorem~\ref{thm:asymptotic-upper-bound}]
Since each $u(j)$ is drawn independently from the uniform distribution on $[0,1]$, the probability that $u(j^*)\geq 1/2$ is $1-1/2^m$, which converges to $1$ for large $m$.
In addition, since $u(j)\in [0.1/k, 0.5/k]$ with probability $0.4/k$ for each $j$, a standard Chernoff bound argument implies that with probability approaching $1$, we have 
\[
M' := |\{j \in M \mid u(j) \in [0.1/k, 0.5/k]\}| \geq 0.3 m/k. 
\]
The union bound implies that both events occur simultaneously with high probability.
Suppose that they both occur and $m \geq 40k^3$.
From the first event, we have $u(j)\leq 0.5/k\leq u(j^*)/k$ for each $j\in M'$, and so $M'\subseteq M_{\text{low-utility}}$.
Hence, the second event implies that
\[
\sum_{j \in M_{\text{low-utility}}} u(j) \geq \sum_{j \in M'} u(j) \geq (0.3m/k)(0.1/k) \geq  k \geq k \cdot u(j^*).
\]
From this and Lemma~\ref{lem:deterministic-condition-upper-bound}, we conclude that with high probability, we can efficiently find a consensus $k$-splitting that cuts only a single item, as claimed.
\end{proof}

\section{Conclusion}

In this paper, we studied a natural version of the consensus halving problem where, in contrast to prior work, the items do not have a linear structure.
We showed that computing a consensus halving with at most $n$ cuts in our version can be done in polynomial time for additive utilities, but already becomes PPAD-hard for a class of monotonic utilities that are very close to additive.
We also demonstrated several extensions and connections to the problems of consensus $k$-splitting and agreeable sets.

While our PPAD-hardness result serves as strong evidence that consensus halving for a set of items is computationally hard for non-additive utilities, it remains open whether the result can be strengthened to PPA-completeness---indeed, the membership of the problem in PPA follows from a reduction to consensus halving on a line, as explained in the introduction. Obtaining a PPA-hardness result will most likely require new ideas and perhaps even new insights into PPA, since all existing PPA-hardness results for consensus halving heavily rely on the linear structure.
Of course, it is also possible that the problem is in fact PPAD-complete.
In addition to consensus halving, settling the computational complexity of the agreeable sets problem for a non-constant number of agents with monotonic utilities would also be of interest.

Another important question that arises from our work is whether there always exists a consensus $k$-splitting with at most $(k-1)n$ cuts when items do not lie on a line and agents have monotonic utilities.
If these utilities are also additive, the claim holds by Theorem~\ref{thm:polytope-k-splitting}; for consensus halving, the claim also holds by reducing to the linear version.
However, for non-additive utilities and unequal ratios, this reduction technique no longer works: even for $k=2$, we may already need to make more than $n$ cuts on the line (Theorem~\ref{thm:k-splitting-more-cuts}).
From a broader point of view, our work illustrates the richness of consensus halving and related problems, which we believe deserve further study.

\section*{Acknowledgments}

This work was partially supported by the European Research Council (ERC) under grant number 639945
(ACCORD), by an EPSRC doctoral studentship (Reference 1892947), and by JST,
ACT-X.

\bibliographystyle{named}
\bibliography{main}

\appendix

\section{Constant Number of Agents}
\label{app:constant-n}

In this section, we provide additional results for the case where there are a constant number of agents who are endowed with monotonic utilities.

\subsection{Discrete Consensus Halving}
\label{app-subsec:discrete-consensus-halving}

We begin by introducing a discrete version of consensus halving, which allows us to focus solely on the agents' utilities for non-fractional sets of items.

\begin{definition}
\label{def:discrete-consensus-halving}
A \emph{discrete consensus halving} is a partition of the items into three (non-fractional) sets of items $M_0,M_1,M_2$ such that $u_i(M_0\cup M_1)\geq u_i(M_2)$ and $u_i(M_0\cup M_2)\geq u_i(M_1)$ for all $i\in N$.
\end{definition}

Note that for any $r$, a consensus halving with $r$ cuts yields a discrete consensus halving with $|M_0|\leq r$ simply by moving all cut items into $M_0$.
Hence, a discrete consensus halving with $|M_0|\leq n$ is guaranteed to exist.
The bound $n$ is also tight here: when each agent values a single distinct item, all valued items must be included in $M_0$.

The following result shows that for constant $n$, a discrete consensus halving with $|M_0|\leq n$ can be found efficiently.
Similarly to Theorem~\ref{thm:agreeable-existence}, the proof involves arranging the items on a line and appealing to the existence of a consensus halving with at most $n$ cuts on the line.
As in Theorem~\ref{thm:agreeable-algo}, we assume a utility oracle model in which the algorithm can query the utility $u_i(M')$ for any $i\in N$ and $M'\subseteq M$.

\begin{theorem}
\label{thm:discrete-consensus-halving-algo}
For any constant number of agents with monotonic utilities, there exists a polynomial-time algorithm that computes a discrete consensus halving with $|M_0|\leq \min\{n,m\}$ (assuming access to a utility oracle).
\end{theorem}

\begin{proof}
If $n\geq m$, we can simply include all items in $M_0$, so assume that $n \leq m$.
Arrange the items on a line in arbitrary order, and extend the utility functions of the agents to fractional sets of items in a continuous and monotonic fashion (see footnote~\ref{fn:monotonic-extension}).
Consider a consensus halving with respect to the extended utilities that uses at most $n$ cuts on the line, and move all cut items to $M_0$.
The resulting discrete consensus halving has the property that for any pair of consecutive items in $M_0$, the block of items in between either all belong to $M_1$ or all belong to $M_2$.

We perform a brute-force search over all possible partitions of the items into $M_0,M_1,M_2$ satisfying the above property.
For $t\in[n]$, there are $O(m^t)$ sets of items that we can choose as $M_0$, and for each choice of $M_0$, there are at most $2^{t+1}$ ways to assign the resulting blocks of items to $M_1$ or $M_2$.
Hence the brute-force search runs in time $\sum_{t=1}^nO(2^{t+1}m^t) = O(m^n)$, which is polynomial since $n$ is constant.
\end{proof}

With two agents, the algorithm in Theorem~\ref{thm:discrete-consensus-halving-algo} runs in quadratic time.
We next present a more sophisticated algorithm that uses only linear time for this special case.
In fact, we will show a stronger statement based on a notion introduced by \citet{KyropoulouSuVo19}.

\begin{definition}
\label{def:Exact1}
Let $n=2$. 
A partition of the items into two (non-fractional) sets of items $M_1$ and $M_2$ is said to be \emph{Exact1} if for each pair $i,k\in\{1,2\}$, either $M_{3-k}=\emptyset$ or there exists an item $j\in M_{3-k}$ such that $u_i(M_k) \geq u_i(M_{3-k}\backslash\{j\})$.
\end{definition}

In words, Exact1 means that for each agent and each part of the partition, this part can be made at least as valuable as the other part in the agent's view by removing at most one item from the latter part.
Given an Exact1 partition, we can easily obtain a discrete consensus halving as follows. From the partition, each agent $i$ proposes (at most) one item to include in $M_0$.
Specifically, if $u_i(M_1)< u_i(M_2)$, then agent $i$ proposes an item $j$ such that $u_i(M_1)\geq u_i(M_2\backslash\{j\})$; the opposite case is analogous. (If $u_i(M_1)=u_i(M_2)$, agent $i$ does not need to propose any item.)
It is clear that $|M_0|\leq 2$, and one can check that $(M_0,M_1,M_2)$ forms a discrete consensus halving.

\citet{KyropoulouSuVo19} showed that an Exact1 partition exists for two agents with ``responsive utilities'', a class that lies between additive and monotonic utilities.
Here, we present an algorithm that computes an Exact1 partition for arbitrary monotonic utilities in linear time---to the best of our knowledge, even the \emph{existence} of such a partition has not been established before.

Our algorithm is based on carefully discretizing a procedure of \citet{Austin82}, which computes a (non-discrete) consensus halving for two agents assuming that the resource is represented by the circumference of a circle.
Austin's procedure works by letting the first agent place two knives on the circle so that the item is cut in half according to her valuation.
The agent then moves both knives continuously clockwise, maintaining the invariant that the knives divide the items into two equal halves in her opinion. 
The first agent stops moving the knives when the two parts are equal according to the valuation of the second agent, and the procedure returns the resulting partition. 
Since the second knife would reach the initial position of the first knife at the same time as the first knife reaches the starting point of the second knife, it follows from the intermediate value theorem that the procedure necessarily terminates.

The main challenge in applying this procedure to our discrete item setting is that it is not a priori clear how to implement moving both knives simultaneously---indeed, moving each of the knives over one item does not always maintain the invariant that the partition is Exact1. 
Nevertheless, as we will show, this invariant can be maintained by either moving both knives or moving one of the two knives, whichever option is appropriate at each stage.
In fact, for this algorithm and proof, we will use a slightly stronger definition of Exact1 wherein the items lie on a circle, each part of the partition forms a contiguous block on the circle, and the item $j$ in Definition~\ref{def:Exact1} is only allowed to be one of the items at the end of block $M_{3-k}$.\footnote{A notion in the same spirit called ``envy-freeness up to one outer good'' was proposed by \citet{BiloCaFl19}.}
We say that a partition is \emph{Exact1 for agent $i$} if the (stronger) Exact1 condition is fulfilled for agent $i$ and both $k\in\{1,2\}$.

\begin{framed}
\noindent
\textbf{Algorithm~1} (for two agents with monotonic utilities) \\

\noindent
\emph{Step~1:} Arrange the items on a circle in arbitrary order.
Place the first knife between two arbitrary consecutive items on the circle, and the second knife between two items so that the partition induced by the two knives is Exact1 for the first agent.  \\

\noindent
\emph{Step~2:} If the current partition is Exact1 for the second agent, return this partition. \\

\noindent
\emph{Step~3:} If one of the knives is at the initial position of the other knife, go to Step~4. 
Else, perform one of the following actions so that the new partition remains Exact1 for the first agent:
\begin{enumerate}[label=(\alph*)]
\item Move the first knife clockwise by one position.
\item Move the second knife clockwise by one position.
\item Move each of the two knives clockwise by one position.
\end{enumerate}
Go back to Step~2. \\

\noindent
\emph{Step~4:} Move the knife that is not at the initial position of the other knife clockwise by one position. Go back to Step~2.
\end{framed}

\begin{theorem}
\label{thm:austin}
For two agents with monotonic utilities, Algorithm~1 computes an Exact1 partition in time linear in $m$ (assuming access to a utility oracle).
\end{theorem}

\begin{proof}
Observe that throughout the algorithm, the partition induced by the two knives is Exact1 for the first agent. 
Moreover, a partition is returned only if it is Exact1 for the second agent. 
Hence, if the algorithm terminates, the partition that it outputs is Exact1 for both agents.
It therefore suffices to establish that the algorithm is well-defined and always terminates.
For convenience, we will say that a bundle is \emph{envy-free up to one item (EF1)} for a specific agent if the Exact1 condition (specifically, the stronger version described before the algorithm) is fulfilled for the agent when that bundle is taken as $M_k$.

First, we need to show that in Step~1, there exists a position of the second knife such that the resulting partition is Exact1 for the first agent.
It turns out that this already follows from Theorem~3.1 of \cite{OhPrSu19}, so the first step can be implemented.

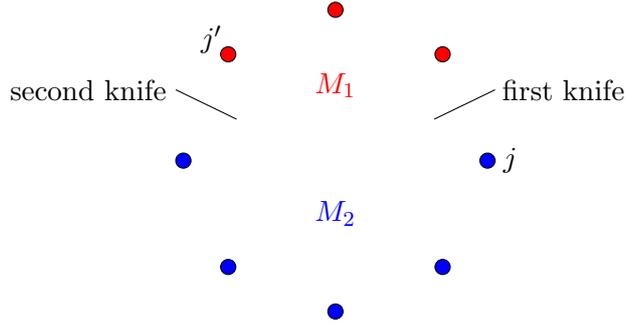
\begin{figure}[!ht]
\centering
\begin{tikzpicture}
\draw [fill=blue] (2,0) circle [radius = 0.1];
\draw [fill=red] (2,4) circle [radius = 0.1];
\draw [fill=blue] (0,2) circle [radius = 0.1];
\draw [fill=blue] (4,2) circle [radius = 0.1];
\draw [fill=red] (3.41,3.41) circle [radius = 0.1];
\draw [fill=red] (0.59,3.41) circle [radius = 0.1];
\draw [fill=blue] (3.41,0.59) circle [radius = 0.1];
\draw [fill=blue] (0.59,0.59) circle [radius = 0.1];
\draw (3.3,2.55) -- (4.1,2.94);
\draw (0.7,2.55) -- (-0.1,2.94);
\node[red] at (2,3) {$M_1$};
\node[blue] at (2,1.3) {$M_2$};
\node at (0.37,3.61) {$j'$};
\node at (4.3,2) {$j$};
\node at (5,2.94) {first knife};
\node at (-1.25,2.94) {second knife};
\end{tikzpicture}
\caption{An illustration of Algorithm~1.}
\label{fig:moving-knife}
\end{figure}

Next, the key part of our proof is to show that in Step~3, at least one of the three actions keeps the new partition Exact1 for the first agent. 
Assume that actions (a) and (b) do not; we claim that action (c) does.
Call the two parts of the partition $M_1$ and $M_2$, and assume without loss of generality that moving the first knife clockwise would enlarge $M_1$.
Suppose that the next item that the first knife would move over is $j$, and the next item that the second knife would move over is $j'$. (See Figure~\ref{fig:moving-knife} for an illustration.)
Let $O_1 = (M_1\cup \{j\})\backslash\{j'\}$ and $O_2 = (M_2\cup \{j'\})\backslash\{j\}$ be the two parts of the partition that results from action (c).
Since action (a) does not keep the partition Exact1, we have that $M_2\backslash\{j\}$ is not EF1 for the first agent.
Hence,
\[u_1(O_1) = u_1((M_1\cup \{j\})\backslash\{j'\}) > u_1(M_2\backslash\{j\}) = u_1(O_2\backslash\{j'\}),\]
implying that $O_1$ is EF1 for the first agent.
By symmetry, since action (b) does not keep the partition Exact1, we have that $M_1\backslash\{j'\}$ is not EF1 for the first agent. This implies that $O_2$ is EF1 for the agent.
It follows that action (c) keeps the partition Exact1 for the first agent, as claimed.

Now, consider Step~4.
Since each knife never moves by more than one position at a time, unless the algorithm terminates beforehand, this step will eventually be reached.
Suppose that the first knife has arrived at the initial position of the second knife, but the second knife is not yet at the initial position of the first knife.
The current partition is Exact1 for the first agent.
Also, if the second knife moves clockwise to the initial position of the first knife, again we have an Exact1 partition for the agent.
Hence, monotonicity of the EF1 property implies that every partition in between is also Exact1 for the agent.

Finally, we show that the algorithm necessarily terminates. 
Suppose that this is not the case.
Assume that in the initial partition with parts $M_1$ and $M_2$, the second agent believes that $M_2$ is not EF1.
This means that $u_2(M_2)<u_2(M_1\backslash\{j\})$ for any $j$ at the end of block $M_1$.
In one iteration of Step~3 or 4, $M_1$ loses at most one end item to $M_2$---call this item $j'$ (if $M_1$ does not lose any item, take $j'$ to be an arbitrary end item in $M_1$), and the respective parts of the partition after the iteration $O_1$ and $O_2$.
Since $u_2(M_1\backslash\{j'\}) > u_2(M_2) = u_2((M_2\cup\{j'\})\backslash\{j'\})$, we have that $O_1$ is also EF1 for the second agent.
However, since the algorithm does not terminate here by assumption, $O_2$ is not EF1 for the second agent.
The same argument tells us that in further iterations, the second bundle (i.e., $M_2$, $O_2$, and so on) is still not EF1 for the agent.
However, the algorithm must reach a point where the first knife is at the initial position of the second knife and, at the same time, the second knife is also at the initial position of the first knife.
At this point, the second bundle coincides with the initial first bundle, so it must be EF1 for the second agent.
This yields the desired contradiction.

Regarding the running time, note that each knife moves clockwise around the circle only once, so the number of partitions considered by the algorithm is linear. 
For each partition, checking the relevant Exact1 condition can be done in constant time since it involves hypothetically removing only a constant number of items.
Hence the algorithm runs in linear time, as claimed.
\end{proof}

\subsection{Continuous Extensions}
\label{app-subsec:continuous-extensions}

The discrete consensus halving problem allows us to concern ourselves exclusively with the agents' utilities for non-fractional sets of items, which are represented by set functions.
For an additive set function, there exists an obvious extension to fractional sets of items: the linear extension used in Section~\ref{sec:additive}.
This is, however, not the case for general monotonic functions.
In this subsection, we address two extensions that have been studied in the literature, namely the Lov\'{a}sz extension and the multilinear extension.
We refer to the lecture notes of \citet{Vondrak10} for further discussion of these extensions.

Let $\textbf{x} = (x_1,\dots,x_m)$, and for each subset $S\subseteq[m]$, denote by $\textbf{1}_S$ the vector of length $m$ such that the $i$th component is $1$ if $i\in S$, and $0$ otherwise.

\begin{definition}
\label{def:lovasz-extension}
Given a function $f:\{0,1\}^m\rightarrow\mathbb{R}$, the \emph{Lov\'{a}sz extension} $f^L:[0,1]^m\rightarrow\mathbb{R}$ of $f$ is defined by
\[
f^L(\textbf{x}) = \sum_{i=0}^m \lambda_i f(S_i),
\]
where $\emptyset = S_0\subset S_1\subset\dots\subset S_m = [m]$ is a chain such that $\sum_{i=0}^m \lambda_i\textbf{1}_{S_i} = \textbf{x}$ for $\lambda_0,\lambda_1,\dots,\lambda_m\geq 0$ with $\sum_{i=0}^m \lambda_i = 1$.
\end{definition}

\noindent As an example, suppose that $m=3$ and $\textbf{x} = (1, 0.1, 0.3)$.
Then we have $S_1=\{1\}$, $S_2=\{1,3\}$, $S_3 = \{1,2,3\}$, and $(\lambda_0,\lambda_1,\lambda_2,\lambda_3) = (0,0.7,0.2,0.1)$, meaning that $$f^L(\textbf{x}) = 0.7\cdot f(\{1\}) + 0.2\cdot f(\{1,3\}) + 0.1\cdot f(\{1,2,3\}).$$

\begin{definition}
\label{def:multilinear-extension}
Given a function $f:\{0,1\}^m\rightarrow\mathbb{R}$, the \emph{multilinear extension} $F:[0,1]^m\rightarrow\mathbb{R}$ of $f$ is defined by
\[
F(\textbf{x}) = \sum_{S\subseteq [m]} f(S) \prod_{i\in S}x_i\prod_{i\in[m]\backslash S}(1-x_i).
\]
\end{definition}

\noindent For the example above, we have
\begin{align*}
F(\textbf{x}) 
&= 0.9\cdot 0.7\cdot f(\{1\})
+ 0.1\cdot 0.7\cdot f(\{1,2\})
+ 0.9\cdot 0.3\cdot f(\{1,3\})
+ 0.1\cdot 0.3\cdot f(\{1,2,3\}) \\
&= 0.63\cdot f(\{1\})
+ 0.07\cdot f(\{1,2\})
+ 0.27\cdot f(\{1,3\})
+ 0.03\cdot f(\{1,2,3\}).
\end{align*}

\citet{Vondrak10} proved that if $f$ is a monotonic set function, then its multilinear extension $F$ is also monotonic, i.e., increasing a component $x_i$ by any amount does not decrease the value of the function $F(\textbf{x})$.
For completeness, we show an analogous result for the Lov\'{a}sz extension.

\begin{proposition}
If a function $f:\{0,1\}^m\rightarrow\mathbb{R}$ is monotonic, then so is its Lov\'{a}sz extension $f^L$.
\end{proposition}

\begin{proof}
Let $f$ be a monotonic set function, and $f^L$ be its Lov\'{a}sz extension.
Let $\textbf{x}\in[0,1]^m$, and assume that $x_1\leq x_2\leq\dots\leq x_m$ (other orderings can be handled analogously). 
In this case, we have
\begin{align*}
f^L(\textbf{x})
=\hspace{0.3em} & x_1f(\{1,2,\dots,m\}) + (x_2-x_1)f(\{2,3,\dots,m\}) + \dots \\
&+ (x_i-x_{i-1})f(\{i,\dots,m\}) + (x_{i+1}-x_i)f(\{i+1,\dots,m\}) + \dots \\
&+ (x_m-x_{m-1})f(\{m\}).
\end{align*}
It suffices to show that for any $i$, the value $f^L(\textbf{x})$ does not decrease upon increasing $x_i$.
This is obvious if $i=m$.
For $1\leq i\leq m-1$, we only need to prove that $f^L(\textbf{x})$ does not decrease when we increase $x_i$ until it reaches $x_{i+1}$---indeed, if we want to increase $x_i$ further, we can swap the roles of $x_i$ and $x_{i+1}$ and apply the same argument.
When we increase $x_i$ in the range $[x_{i-1},x_{i+1}]$, the only terms that change are $x_i\cdot f(\{i,\dots,m\})$ and $-x_i\cdot f(\{i+1,\dots,m\})$.
The net change is
\[
x_i\cdot (f(\{i,\dots,m\}) - f(\{i+1,\dots,m\})),
\]
which is nonnegative due to the monotonicity of $f$.
The conclusion follows.
\end{proof}

When $n$ is constant, computing a consensus halving for a utility function given by the Lov\'{a}sz extension of a monotonic set function can be done efficiently.

\begin{theorem}
\label{thm:lovasz-extension-algo}
For a constant number of agents with monotonic utilities, each given by the Lov\'{a}sz extension of a set function, there exists a polynomial-time algorithm that computes a consensus halving with at most $\min\{n,m\}$ cuts (assuming access to a utility oracle for the set function).
\end{theorem}

\begin{proof}
If $n\geq m$, we can simply divide every item in half, so assume that $n\leq m$. 
Arrange the items on a line in arbitrary order.
Similarly to the proof of Theorem~\ref{thm:discrete-consensus-halving-algo}, there exists a consensus halving that uses at most $n$ cuts on the line such that for any pair of consecutive cut items, the block of whole items in between either all belong to $M_1$ or all belong to $M_2$.
We will perform a brute-force search over all partitions of items into $(M_0,M_1,M_2)$ such that all cut items belong to $M_0$ and the above property is satisfied; as in Theorem~\ref{thm:discrete-consensus-halving-algo}, this search takes polynomial time.

For each such partition, it remains to determine the ratios by which we should divide the items in $M_0$ between $M_1$ and $M_2$.
Denote by $x_1,\dots,x_r$ the fraction of the $r\leq n$ items in $M_0$ that should go into $M_1$.
We iterate over all possible orderings of $x_1,\dots,x_r$---there are at most $n!$ orderings, which is polynomial since $n$ is constant.
For each ordering, one can verify that the consensus-halving condition for each agent reduces to a linear equation in $x_1,\dots,x_r$.
Hence, to check the feasibility of a partition along with an ordering, we can run any efficient linear programming algorithm (with an arbitrary objective) on the ordering and consensus-halving constraints.
The previous paragraph implies that at least one combination of partition and ordering results in a feasible linear program, which in turn gives rise to the desired consensus halving.
\end{proof}

A consequence of Theorem~\ref{thm:lovasz-extension-algo} is that for the Lov\'{a}sz extension, if the set function is rational, then there exists a consensus halving with rational ratios.
By contrast, for the multilinear extension, a consensus halving may necessarily involve splitting items in irrational ratios, even if the set function only takes on integer values.

\begin{theorem}
\label{thm:multilinear-irrational}
There exists an instance with $n=2$ and $m=3$ in which each agent has a monotonic utility function given by the multilinear extension of a set function taking on integer values, but every consensus halving with at most two cuts involves splitting some items in irrational ratios.
\end{theorem}

\begin{proof}
Assume that $n=2$ and $m=3$. The utility functions of the agents are given in Table~\ref{table:utilities-example}.
Notice that the function of the second agent is the same as that of the first agent, except with the roles of items $2$ and $3$ reversed.

\begin{table}
\centering
\begin{tabular}{ | c | c | c | }
  \hline
  Set $S$ & $u_1(S)$ & $u_2(S)$ \\ \hline 
  \hline
  $\emptyset$ & $0$ & $0$ \\ \hline
  $\{1\}$ & $1$ & $1$  \\ \hline
  $\{2\}$ & $10$ & $2$ \\ \hline
  $\{3\}$ & $2$ & $10$ \\ \hline
  $\{1,2\}$ & $12$ & $3$ \\ \hline
  $\{1,3\}$ & $3$ & $12$ \\ \hline
  $\{2,3\}$ & $14$ & $14$ \\ \hline
  $\{1,2,3\}$ & $20$ & $20$ \\ 
  \hline
\end{tabular}
\caption{Utility functions for the instance in the proof of Theorem~\ref{thm:multilinear-irrational}.}
\label{table:utilities-example}
\end{table}

Consider a consensus halving $(M_1,M_2)$ of this instance with at most two cuts.
Since $u_1(2) > u_1(\{1,3\})$, item $2$ needs to be cut.
Similarly, since $u_2(3) > u_2(\{1,2\})$, item $3$ needs to be cut.
Hence item $1$ must be uncut; assume without loss of generality that it belongs to $M_1$.
Let $x_2$ and $x_3$ be the fraction of item $2$ and $3$ in $M_1$, respectively.
Since $u_1(M_1) = u_1(M_2)$, we have
\begin{align*}
&(1-x_2)(1-x_3)\cdot u_1(1) + x_2(1-x_3)\cdot u_1(\{1,2\}) + x_3(1-x_2)\cdot u_1(\{1,3\}) + x_2x_3\cdot u_1(\{1,2,3\}) \\
&= x_2x_3\cdot u_1(\emptyset) + x_3(1-x_2)\cdot u_1(2) + x_2(1-x_3)\cdot u_1(3) + (1-x_2)(1-x_3)\cdot u_1(\{2,3\}).
\end{align*}
This is equivalent to
\[
(1-x_2)(1-x_3)\cdot (-13) + x_2(1-x_3)\cdot 10 + x_3(1-x_2)\cdot (-7) + x_2x_3\cdot 20 = 0,
\]
or 
\begin{equation}
-13+23x_2+6x_3+4x_2x_3 = 0.
\label{eq:agent-1}
\end{equation}
By symmetry, $u_2(M_1) = u_2(M_2)$ implies that 
\begin{equation}
-13+23x_3+6x_2+4x_2x_3 = 0.
\label{eq:agent-2}
\end{equation}
Subtracting \eqref{eq:agent-2} from \eqref{eq:agent-1} yields $17x_2 = 17x_3$, so $x_2=x_3$.
Plugging this back into \eqref{eq:agent-1}, we get 
\begin{equation}
4x_2^2+29x_2-13 = 0.    
\label{eq:quadratic}
\end{equation}
The only positive solution to \eqref{eq:quadratic} is $x_2 = \frac{\sqrt{1049}-29}{8}\approx 0.4235\dots$, meaning that every consensus halving involves splitting items $2$ and $3$ in irrational ratios.
\end{proof}

Theorem~\ref{thm:multilinear-irrational} implies that for the multilinear extension, computing a consensus halving exactly may not be possible if our computation model only allows representing rational numbers.
As we can see, with two agents and two necessary cuts, the problem already requires solving a quadratic equation.
For more agents, we can therefore expect that one would need to solve higher-degree polynomial equations---the Abel-Ruffini theorem states that almost all polynomials of degree at least five do not admit a solution in radicals.
Hence, for this extension, finding an approximate consensus halving is likely the best that one could do even under general computational models.

\section{Proof of Lemma~\ref{lem:simple-g-circuit}}
\label{app:proof-g-circuit}

We reduce from the \textsc{$\varepsilon$-Gcircuit} problem, which is known to be PPAD-hard even for some constant $\varepsilon > 0$ \citep{Rubinstein18}. In this problem we are given a generalized circuit $(V,\mathcal{T})$, where there are $9$ gate types: $G_\zeta$, $G_{\times \zeta}$, $G_=$, $G_+$, $G_-$, $G_<$, $G_\lor$, $G_\land$ and $G_\lnot$
with $\zeta\in[0,1]$ for the first two gates (see \citep{Rubinstein18} for a formal definition of the gates). The last three gate types correspond to Boolean operations. As shown by \citet[Corollary 1]{SchuldenzuckerSe19}, these three gate types are actually not necessary, and the problem remains PPAD-hard for constant $\varepsilon$ even without them. Apart from the set of gates, the other difference with \textsc{$\varepsilon$-simple-Gcircuit} is that in \textsc{$\varepsilon$-Gcircuit} we want to assign a number in $[0,1]$ to each node (instead of $[-1,1]$).
	
Let $\widehat{\varepsilon} > 0$ be a constant such that the \textsc{$\widehat{\varepsilon}$-Gcircuit} problem without Boolean operation gates is PPAD-hard, and let $(V,\mathcal{T})$ be an instance of \textsc{$\widehat{\varepsilon}$-Gcircuit} without Boolean gates. We construct an instance $(V',\mathcal{T}')$ of \textsc{$\varepsilon$-simple-Gcircuit}, where $\varepsilon>0$ is a sufficiently small constant (which we pick later), such that any solution to the new instance yields a solution to the original instance. We let $V' = V \cup V_{aux}$, where $V_{aux}$ is a set of nodes that will be used for ``intermediate'' results when simulating the gates of the original problem with the restricted set of gates allowed in \textsc{$\varepsilon$-simple-Gcircuit}. We will construct $\mathcal{T}'$ such that it induces the original constraints of $\mathcal{T}$ on the nodes $V \subset V'$. Furthermore, we will also ensure that in any solution $\xx : V' \to [-1,1]$, we have $\xx[v] \in [0,1]$ for all $v \in V \subset V'$. Thus, restricting $\xx$ to $V$ will immediately yield a solution to the original \textsc{$\widehat{\varepsilon}$-Gcircuit} instance.
	
Recall that we only have three types of gates at our disposal: $G_+$, $G_{\times -|\zeta|}$ for $\zeta \in (0,1]$, and $G_1$. We begin by constructing some useful gadgets that simulate more operations on the same domain $[-1,1]$. Throughout, we denote the input nodes by $u_1,u_2$ (if applicable) and the output node by $v$.
	
\paragraph{$G_{\times \zeta}$: multiplication by $\zeta \in [-1,1]$.} This gadget ensures that $\xx[v] = \zeta \cdot \xx[u_1] \pm 2\varepsilon$. If $\zeta < 0$, use a $G_{\times -|\zeta|}$ gate with input $u_1$ and output $v$. If $\zeta > 0$, use a $G_{\times -|\zeta|}$ gate with input $u_1$ and output $w \in V_{aux}$, and then a $G_{\times -|1|}$ gate with input $w$ and output $v$, which ensures that $\xx[v] = \zeta \cdot \xx[u_1] \pm 2\varepsilon$. Finally, if $\zeta = 0$, then use a $G_{\times -|1|}$ gate with input $u_1$ and output $w \in V_{aux}$, and then a $G_+$ gate with inputs $u_1,w$ and output $v$. This ensures that $\xx[v] = 0 \pm 2\varepsilon$.
	
\paragraph{$G_\zeta$: constant $\zeta \in [-1,1]$.} This gadget ensures that $\xx[v] = \zeta \pm 3\varepsilon$. We use a $G_1$ gate with output $w \in V_{aux}$, and then a $G_{\times \zeta}$ gadget with input $w$ and output $v$, which yields the desired result.

\paragraph{$G_{\times 2}$: multiplication by $2$.} This gadget ensures that $\xx[v] = T_{[-1,1]}(2 \xx[u_1]) \pm 3\varepsilon$. We use a $G_{\times 1}$ gadget with input $u_1$ and output $w \in V_{aux}$, and then a $G_+$ gate with inputs $u_1,w$ and output $v$, which yields the desired result.

\bigskip
	
Before we show how to construct gadgets that simulate the gates of \textsc{$\widehat{\varepsilon}$-Gcircuit}, we need a way to ensure that for $v \in V \subset V'$, we have $\xx[v] \in [0,1]$. To achieve this we will make extensive use of the following gadget.
	
\paragraph{$G_{[0,1]}$: truncation to $[0,1]$.} This gadget ensures that $\xx[v] \in [0,1]$ and $\xx[v] = T_{[0,1]}(\xx[u_1]) \pm 16\varepsilon$. To achieve this we use the fact that for any $t \in [-1,1]$, it holds that $T_{[0,1]}(t) = T_{[-1,1]}[t + (-1)] + 1$. First, we use a $G_{-1}$ gadget with output $w_1 \in V_{aux}$, and then a $G_+$ gate with inputs $u_1,w_1$ and output $w_2 \in V_{aux}$. Next, we use a $G_1$ gate with output $w_3 \in V_{aux}$, and then a $G_+$ gate with inputs $w_2,w_3$ and output $w_4 \in V_{aux}$. Since, the $G_{-1}$ gadget has error at most $3\varepsilon$ and the $G_+$ and $G_1$ gates have error at most $\varepsilon$, we obtain that $\xx[w_4] = T_{[0,1]}(\xx[u_1]) \pm 6\varepsilon$. Furthermore, it holds that $\xx[w_4] \geq -2\varepsilon$, since $\xx[w_4] = T_{[-1,1]}(\xx[w_2]+\xx[w_3]) \pm \varepsilon$, $\xx[w_2] \in [-1,1]$ and $\xx[w_3] \geq 1 - \varepsilon$. Finally, we also use a $G_{6\varepsilon}$ gadget with output $w_5 \in V_{aux}$, and a $G_+$ gate with inputs $w_4,w_5$ and output $v$. This introduces an additional error of at most $4\varepsilon$ and thus ensures that $\xx[v] = T_{[-1,1]}(T_{[0,1]}(\xx[u_1]) + 6\varepsilon) \pm 10\varepsilon = T_{[0,1]}(\xx[u_1]) \pm 16\varepsilon$. Furthermore, it also holds that $\xx[v] \geq T_{[-1,1]}(\xx[w_4] + \xx[w_5]) - \varepsilon \geq 0$, since $\xx[w_4] \geq -2\varepsilon$ and $\xx[w_5] \geq 6\varepsilon - 3\varepsilon$.

\bigskip
	
We are now ready to simulate the constraints $\mathcal{T}$ of the original instance on the nodes $V \subset V'$. First of all, for any node $v \in V$ that does not appear as the output of any gate in $\mathcal{T}$, we ensure that $\xx[v] \in [0,1]$ as follows: create a node $w \in V_{aux}$ and use a $G_{[0,1]}$ gadget with input $w$ and output $v$. Note that we do not care about the error in this case, since we only want to ensure that $\xx[v] \in [0,1]$. For all $v \in V$ that appear as the output of some gate in $\mathcal{T}$, the gadget that outputs into $v$ will ensure that $\xx[v] \in [0,1]$.
	
For every gate $T=(G,u_1,u_2,v,\zeta) \in \mathcal{T}$, we ensure that the corresponding constraint holds as follows:
	
\paragraph{$(G_\zeta^{[0,1]}, nil,nil,v,\zeta)$: constant $\zeta \in [0,1]$.} This gadget ensures that $\xx[v] \in [0,1]$ and $\xx[v] = \zeta \pm 19\varepsilon$. We use a $G_\zeta$ gadget with output $w \in V_{aux}$, and then a $G_{[0,1]}$ gadget with input $w$ and output $v$.
	
\paragraph{$(G_{\times \zeta}^{[0,1]},u_1,nil,v,\zeta)$: multiplication by $\zeta \in [0,1]$.} This gadget ensures that $\xx[v] \in [0,1]$ and $\xx[v] = T_{[0,1]}(\zeta \cdot \xx[u_1]) \pm 18\varepsilon$. We use a $G_{\times \zeta}$ gadget with input $u_1$ and output $w \in V_{aux}$, and then a $G_{[0,1]}$ gadget with input $w$ and output $v$.

\paragraph{$(G_=^{[0,1]},u_1,nil,v,nil)$: copy.} This gadget ensures that $\xx[v] \in [0,1]$ and $\xx[v] = T_{[0,1]}(\xx[u_1]) \pm 16\varepsilon$. For this we simply use the $G_{[0,1]}$ gadget with input $u_1$ and output $v$.
	
\paragraph{$(G_+^{[0,1]},u_1,u_2,v,nil)$: addition.} This gadget ensures that $\xx[v] \in [0,1]$ and $\xx[v] = T_{[0,1]}(\xx[u_1] + \xx[u_2]) \pm 17\varepsilon$. We use a $G_+$ gate with inputs $u_1,u_2$ and output $w \in V_{aux}$, and then a $G_{[0,1]}$ gadget with input $w$ and output $v$.
	
\paragraph{$(G_-^{[0,1]},u_1,u_2,v,nil)$: subtraction.} This gadget ensures that $\xx[v] \in [0,1]$ and $\xx[v] = T_{[0,1]}(\xx[u_1] - \xx[u_2]) \pm 18\varepsilon$. We use a $G_{\times -|1|}$ gate with input $u_2$ and output $w_1 \in V_{aux}$, then a $G_+$ gate with inputs $u_1,w_1$ and output $w_2 \in V_{aux}$, and finally a $G_{[0,1]}$ gadget with input $w_2$ and output $v$.
	
\paragraph{$(G_<^{[0,1]},u_1,u_2,v,nil)$: comparison.} This gate ensures that $\xx[v] \in [0,1]$ and
\begin{itemize}
	\item if $\xx[u_1] < \xx[u_2] - \widehat{\varepsilon}$, then $\xx[v] = 1 \pm 19\varepsilon$,
	\item if $\xx[u_1] > \xx[u_2] + \widehat{\varepsilon}$, then $\xx[v] = 0 \pm 19\varepsilon$.
\end{itemize}
We use a $G_{\times -|1|}$ gate with input $u_1$ and output $w \in V_{aux}$, and then a $G_+$ gate with inputs $u_2,w$ and output $w_0 \in V_{aux}$. This ensures that $\xx[w_0] = T_{[-1,1]}(\xx[u_2] - \xx[u_1]) \pm 2\varepsilon$. Let $k = \lceil \log_2 (1/\widehat{\varepsilon}) \rceil + 1$, so $2/\widehat{\varepsilon}\leq 2^k\leq 4/\widehat{\varepsilon}$. Next, for $i \in [k]$, we use a $G_{\times 2}$ gadget with input $w_{i-1}$ and output $w_i \in V_{aux}$. Finally, we use a $G_{[0,1]}$ gadget with input $w_k$ and output $v$.
	
For the analysis, we first consider the case $\xx[u_1] < \xx[u_2] - \widehat{\varepsilon}$. Then, it holds that $\xx[w_0] \geq \widehat{\varepsilon} - 2\varepsilon \geq \widehat{\varepsilon} - 3\varepsilon$. First, let us show by contradiction that there must exist $i \in [k]$ such that $\xx[w_i] < 2 \xx[w_{i-1}] - 3\varepsilon$. Assume that for all $i \in [k]$ we have $\xx[w_i] \geq 2 \xx[w_{i-1}] - 3\varepsilon$. Then it follows that $\xx[w_k] \geq 2^k \widehat{\varepsilon} - 3\varepsilon \sum_{i=0}^k 2^i \geq 2^k \widehat{\varepsilon} - 3\varepsilon 2^{k+1} \geq 2\widehat{\varepsilon}/\widehat{\varepsilon} -3\varepsilon \cdot 8/\widehat{\varepsilon} = 2 - 24\varepsilon/\widehat{\varepsilon}$. If we ensure that $\varepsilon < \widehat{\varepsilon}/24$, then we obtain that $\xx[w_k] > 1$, which is impossible. Thus, let $i \in [k]$ be such that $\xx[w_i] < 2 \xx[w_{i-1}] - 3\varepsilon$. Recall that the $G_{\times 2}$ gadget ensures that $\xx[w_i] \geq T_{[-1,1]}(2\xx[w_{i-1}]) - 3\varepsilon$. Thus, it must be that $T_{[-1,1]}(2 \xx[w_{i-1}]) < 2 \xx[w_{i-1}]$, i.e., $2 \xx[w_{i-1}] \geq 1$. This implies that $\xx[w_i] \geq 1 - 3\varepsilon$ and thus $\xx[w_{i+1}] \geq T_{[-1,1]}(2-6\varepsilon) - 3\varepsilon \geq 1 - 3\varepsilon$, if we ensure that $\varepsilon \leq 1/6$. By induction, it follows that $\xx[w_k] \geq 1 - 3\varepsilon$ and thus $\xx[v] \geq 1 - 19\varepsilon$, as desired.
	
Now, consider the case $\xx[u_1] > \xx[u_2] + \widehat{\varepsilon}$. By an analogous argument, we obtain that it must be that $\xx[w_k] \leq -1 + 3\varepsilon$, and thus $\xx[v] = T_{[0,1]}(\xx[w_k]) \pm 16\varepsilon = 0 \pm 19\varepsilon$.
	
\bigskip
	
We can now finish the reduction. We set $\varepsilon = \widehat{\varepsilon}/25$. This ensures that all the assumptions we have made about $\varepsilon$ hold, and that all the gadgets that simulate the gates in $\mathcal{T}$ have error at most $\widehat{\varepsilon}$.

\end{document}